\documentclass[acmsmall,nonacm]{acmart}

\AtBeginDocument{%
  }

\usepackage[appendix=append,bibliography=common]{apxproof}
\usepackage{hyperref}
\usepackage{enumerate,scalerel,bussproofs,subcaption,tikz,stmaryrd}
\usetikzlibrary{automata,backgrounds,shapes,arrows,positioning,calc,fit}
\usetikzlibrary{decorations,decorations.pathreplacing,decorations.pathmorphing,arrows.meta}

\newtheoremrep{theorem}{Theorem}[section]
\newtheoremrep{lemma}[theorem]{Lemma}
\newtheoremrep{example}[theorem]{Example}
\newtheoremrep{defn}[theorem]{Definition}

\newcommand\naturals{\mathbb{N}}
\newcommand\rationals{\mathbb{Q}}
\newcommand\reals{\mathbb{R}}
\newcommand\len[1]{|#1|}
\newcommand\Iff\Leftrightarrow
\newcommand\Imp\Rightarrow
\newcommand\set[2]{\{#1\mid#2\}}
\newcommand\subs\subseteq
\newcommand\eps\varepsilon
\newcommand\del\partial
\newcommand\subst[2]{[#1/#2]}
\newcommand\seq[3]{#1_{#2},\ldots,#1_{#3}}
\newcommand\lam[2]{\lambda#1\kern1pt.\kern1pt#2}
\newcommand\id{\mathsf{id}}

\newcommand\powerset[1]{2^{#1}}
\newcommand\sem[1]{\llbracket#1\rrbracket}
\newcommand\semr[1]{\sem{#1}\rho}
\newcommand\amp{\mathrel\&}
\newcommand\cmp{\mathrel;}
\newcommand\skp{\mathsf{skip}}
\newcommand\fail{\mathsf{fail}}
\newcommand\dist{\mathop{\mathsf{dist}}}
\newcommand\letin[3]{\mathsf{let\ }#1=#2\mathsf{\ in\ }#3}
\newcommand\sample{\mathop{\mathsf{sample}}}
\newcommand\rest\restriction
\newcommand\Exp{\mathsf{Exp}}
\newcommand\supp{\mathop{\mathsf{supp}}}
 
\DeclareMathOperator*{\bigamp}{\scalerel*{\amp}{\textstyle\sum}} 
\newcommand\fix[2]{\mathsf{fix}\,#1\kern1pt.\kern1pt#2}
\newcommand\bind{\mathrel{\texttt{>\kern-0.5pt>=}}}

\newcommand\Mfin{M_{\mathrm{fin}}}
\newcommand\eval{\mathsf{eval}}
\newcommand\mset[1]{\{\kern-2pt|#1|\kern-2pt\}}

\newcommand\Brz{\mathsf{Brz}}
\newcommand\ieq{\mathrel{\phantom{=}}}

\newcommand\NX{\naturals^X}
\newcommand\Ns{\naturals^S}
\newcommand\NS{\naturals^{\Sigma^*}}
\newcommand\DNS{D(\NS)}

\newcommand\dangle[1]{\angle{\kern-2pt\angle#1\kern-2pt}}

\renewcommand\angle[1]{\langle#1\rangle}

\newcommand\dirac[1]{\delta_{#1}}

\newcommand\Var{\mathsf{Var}}
\newcommand\half{\textstyle\frac 12}

\newcommand\auto{h}
\newcommand\rnd{\mathsf{rnd}}
\renewcommand\emptyset\varnothing
\newcommand\ifelse[3]{\textsf{if } #1 \textsf{ then } #2 \textsf{ else } #3}
\newcommand\while[2]{\textsf{while } #1 \textsf{ do } #2}

\begin{document}

\title[Probability and Angelic Nondeterminism with Multiset Semantics]{Probability and Angelic Nondeterminism\\with Multiset Semantics}
\subtitle{Extended version of ``Probabilistic Kleene Algebra with Angelic Nondeterminism,'' PLDI 2025 \cite{OMK25a}}

\author{Shawn Ong}
\affiliation{%
  \institution{Cornell University}
  \city{Ithaca}
  \state{New York}
  \country{USA}}
\email{so396@cornell.edu}

\author{Stephanie Ma}
\affiliation{%
  \institution{Cornell University}
  \city{Ithaca}
  \state{New York}
  \country{USA}}
\email{ym363@cornell.edu}

\author{Dexter Kozen}
\affiliation{%
  \institution{Cornell University}
  \city{Ithaca}
  \state{New York}
  \country{USA}}
\email{kozen@cs.cornell.edu}


\begin{abstract}
We introduce a version of probabilistic Kleene algebra with angelic nondeterminism and a corresponding class of automata. Our approach implements semantics via distributions over multisets in order to overcome theoretical barriers arising from the lack of a distributive law between the powerset and Giry monads. We produce a full Kleene theorem and a coalgebraic theory, as well as both operational and denotational semantics and equational reasoning principles.
\end{abstract}


\maketitle

\section{Introduction}

\nocite{AbramskyJung94,SKFKS17a}

The combination of probability and nondeterminism in state-based systems is notoriously challenging, chiefly due to the nonexistence of a suitable distributive law between the powerset and probability monads \cite{VaraccaWinskel06,ZwartMarsden22}, leading to a variety of workarounds \cite{Affeldt21,ChenSanders09,DahlqvistParlantSilva18,GoyPetrisan20,HartogdeVink99,KeimelPlotkin17,Mislove00,MisloveOuaknineWorrell04,Varacca03,VaraccaWinskel06,WangHoffmannReps19,Zwart20,ZwartMarsden22}. Most of these approaches treat nondeterminism \emph{demonically}, meaning that the nondeterminism is resolved by an adversarial agent external to the program.

In this paper, we present a version of probabilistic Kleene algebra with angelic nondeterminism and a corresponding automaton model. This is a step in the development of a version of Kleene algebra with tests (KAT) \cite{K97c} with constructs to allow equational reasoning about programs in the presence of probability and nondeterminism. KAT has been quite successful with many variants and spinoffs, with applications in program analysis, networking, verification of compiler optimizations, and many other areas. KAT interprets nondeterminism angelically, which allows KAT-like systems to encode imperative programming constructs such as conditional tests and while loops, thereby reducing the reasoning process to a few basic constructs and equational rules. KAT and its variants are based on Kleene algebra, which forms the core, and that is what we focus on here.

A common but erroneous view of angelic nondeterminism involves a single agent with a stash of random bits making nondeterministic choices at nondeterministic choice states and consulting its stash of random bits at probabilistic choice states. In this approach, the same random bit may be used at different points in the computation, depending on some previous nondeterministic choice of the agent. Naively, it seems that this should cause no problem, since the bit is used only once, depending on which nondeterministic path was taken. However, to work mathematically, all such probabilistic choices should be independent. Rather than a single agent, we should instead think of a collection of agents visiting various states at various points in time, all acting independently. An agent visiting a probabilistic choice state chooses randomly which state to visit next and moves to that state. However, an agent visiting a nondeterministic choice state spawns multiple agents that go separate ways and thereafter act independently. At any point in the computation, there may be several distinct agents visiting the same probabilistic choice state, but all their choices at that state will be independent. Whereas the use of powersets would lose information by conflating these independent agents, the use of multisets retains it.

We take inspiration from the recently established Beck distributive law of probability over finite multisets \cite{Jacobs21,DashStaton21a,DashStaton21b,Dash23} to give a full treatment of probabilistic automata and expressions with angelic nondeterminism formalized by multisets instead of powersets. Our main results are operational and fully compositional denotational semantics, a full Kleene theorem, and a detailed development of the coalgebraic theory, including a ``fundamental theorem'' \`a la \cite{Silva10}. We develop reasoning principles in both denotational and operational styles.

A novel aspect of our approach is that automata and expressions are interpreted over $\DNS$, the space of distributions over multisets of strings with finite multiplicities. Thus a string is not just accepted with some probability, but accepted \emph{with some finite multiplicity} with some probability. We view a language not as a set of strings, but as a multiset of strings.

\subsection{Related Work}

There have been numerous previous attempts at building computational models combining probability and nondeterminism.
A compendium (Table~\ref{fig:Ana}), reproduced from \cite{Sokolova11}, summarizes various models of computation and their corresponding functors when viewed coalgebraically, including many of these approaches.
However, to circumvent theoretical difficulties based on the lack of a suitable distributive law involving the probabilistic and powerset monads \cite{VaraccaWinskel06,Zwart20,ZwartMarsden22}, many of these models must implement some workaround. Some general theoretical approaches include the geometrically convex monad \cite{Affeldt21}, monad lifting \cite{DahlqvistParlantSilva18}, weak distributive laws \cite{GoyPetrisan20,Varacca03,VaraccaWinskel06}, and Kegelspitzen \cite{KeimelPlotkin17,Rennela16}.
In the cases where weakenings of distributive laws are used, the resulting implementations would also have weakened versions of the distributive axioms. This also appears in the geometrically convex monad, probabilistic expressions and powerdomains \cite{MisloveOuaknineWorrell04}, and Kegelspitzen, all of which restrict distributivity in some way.

Previous approaches modeling nondeterminism demonically, such as \cite{MciverCohenMorgan06,MGCM08,Segala95,ZKSTa24}, often require giving an adversary extra nonconstructive power such as the ability to choose any convex combination of allowed possibilities. \emph{Angelic} nondeterminism, on the other hand, is under the control of the programmer in the form of a nondeterministic choice operator ($\amp$) with a well-defined compositional semantics, enabling KAT-style equational reasoning.

Several known operational models have associated algebraic systems that are powerful enough to reason about arbitrary instances of the operational model. However, in these cases, the algebraic model may be more expressive than the operational model, or there is at least no proof that this is not so.
This includes interactive Markov chains \cite{Hermanns02}, Segala systems \cite{SegalaPhD95,SegalaLynch94}, and Markov automata \cite{Hatefi17}.

In other cases, a formal grammar is produced with an operational model with instances equivalent to any admitted by the grammar. Then the opposite problem may occur---the grammar may not be expressive enough to capture all instances of the operational model. This is the case with PKAT expressions and probabilistic configuration transition systems in \cite{QWGW08} and PKAT expressions and the corresponding automaton model in \cite{MRS11}.

Some implementations also restrict various aspects of functionality, such as composition, distributivity, or iteration. These include instances such as weak Kleene algebra \cite{MGCM08}, in which certain distributive axioms are weakened in order to accommodate the lack of a distributive law, PCCS \cite{HartogdeVink99} in which the arguments of parallel composition are restricted to avoid distributivity (additionally, external stacks are used to facilitate backtracking for defining semantics), and PCSP \cite{Mislove00} in which probability does not distribute over nondeterminism.

These attempts include various probabilistic versions of Kleene algebra \cite{MciverCohenMorgan06,QWGW08,MGCM08,MRS11,FTN08,FKMRS16}. However, to our knowledge, none have yielded a Kleene theorem in both directions with full distributivity and iteration.
Our approach is fully compositional and the semantics is defined coinductively without reference to backtracking or additional constructs. Furthermore, it avoids the convoluted constructions that can arise in an effort to deal with the lack of a distributive law. Such constructs appear in the relational model of \cite{ChenSanders09} which adds additional arguments to keep track of the most recent nondeterministic choice, bundle systems \cite{DArgenio98} which require a product construction to implement parallel composition, and control-flow hypergraphs \cite{WangHoffmannReps19} which add call stacks, again to handle backtracking.

\begin{table}
  \caption{Discrete probabilistic system types (from \cite{Sokolova11})}
  \label{fig:Ana}
\begin{tabular}{c@{\hspace{1mm}}c@{\hspace{2mm}}l}
    \toprule
    $\mathsf{Coalg}_F$ & Functor $F$ & \text{Reference}\\
    \midrule
$\mathbf{MC}$ & $D$ & \text{Markov chains}\\
$\mathbf{DLTS}$ & $((-) + 1)^\Sigma$ & \text{deterministic automata}\\
$\mathbf{LTS}$ & $P(\Sigma\times(-))\cong P^\Sigma$ & \text{nondeterministic automata, LTSs}\\
$\mathbf{React}$ & $(D+1)^\Sigma$ & \text{reactive systems \cite{LarsenSkou91,Glabbeek90}}\\
$\mathbf{Gen}$ & $D(\Sigma\times(-))+1$ & \text{generative systems \cite{Glabbeek90}}\\
$\mathbf{Str}$ & $D + (\Sigma\times(-))$ & \text{stratified systems \cite{Glabbeek90}}\\
$\mathbf{Alt}$ & $D + P(\Sigma\times(-))$ & \text{alternating systems \cite{Hansson94}}\\
$\mathbf{Var}$ & $D(\Sigma\times(-)) + P(\Sigma\times(-))$ & \text{Vardi systems \cite{Vardi85}}\\
$\mathbf{SSeg}$ & $P(\Sigma \times D)$ & \text{simple Segala systems \cite{SegalaPhD95,SegalaLynch94}}\\
$\mathbf{Seg}$ & $PD(\Sigma\times(-))$ & \text{Segala systems \cite{SegalaPhD95,SegalaLynch94}}\\
$\mathbf{MA}$ & $P(\Sigma\times D) + P(\rationals \times (-))$ & \text{Markov automata \cite{Hatefi17}}\\
$\mathbf{Bun}$ & $DP(\Sigma\times(-))$ & \text{bundle systems \cite{DArgenio98}}\\
$\mathbf{PZ}$ & $PDP(\Sigma\times(-))$ & \text{Pnueli-Zuck systems \cite{PnueliZuck93}}\\
$\mathbf{MG}$ & $PDP(\Sigma\times(-) + (-))$ & \text{most general systems}\\
  \bottomrule
\end{tabular}
\end{table}

\subsection{Roadmap and Contributions}

\begin{itemize}
\item
In \S\ref{sec:basics}, we discuss the basic mathematical constructs used in this paper, including a review of the recently established Beck distributive law $\otimes:MD\to DM$ of distributions over finite multisets \cite{Jacobs21,DashStaton21a,DashStaton21b,Dash23}.
\item
In \S\ref{sec:automataandexpressions} we introduce automata with probability and angelic nondeterminism and a corresponding language of expressions analogous to regular expressions. We give the denotational semantics of automata and expressions and operational intuition, along with several examples. The semantics of expressions is fully compositional.
\item
In \S\ref{sec:axioms} we list several sound equational axioms along with some examples of their
use, with soundness proofs in \S\ref{sec:Apka}.
\item
In \S\ref{sec:metric} we introduce a complete ultrametric on the space of behaviors $\DNS$ and show that the semantic definitions of automata and expressions give rise to contractive maps, ensuring that the semantic maps of both models are well defined.
\item
In \S\ref{sec:kleene} we give a full Kleene theorem, showing that automata and expressions are equivalent in expressive power. To our knowledge this is the first result of its type for models combining probability and nondeterminism.
\item
In \S\ref{sec:coalgebras} we develop the foundations of the coalgebraic theory, including a notion of Brzozowski derivative and a fundamental theorem \`a la \cite{Silva10}. The usual diagram denoting a unique coalgebra morphism to a final coalgebra turns out not to be appropriate; it is replaced by a more general coalgebra/algebra diagram allowing a recursive definition of a unique map to the space of behaviors $\DNS$.
\item
In \S\ref{sec:future} we discuss future work.
\end{itemize}

Missing proofs and further explanatory material can be found in the
appendix.

\section{Basics}
\label{sec:basics}

For $X$ a set, let $MX = (\naturals\cup\{\infty\})^X$, the set of multisets of $X$ with finite or infinite multiplicities. We use stylized braces $\mset-$ for multiset comprehension. If $f:X\to Y$, then $Mf:MX\to MY$ is the function $Mf(m)(y)=\sum_{f(x)=y}m(x)$. The functor $M$ carries a monad structure with multiset union (pointwise sum) as multiplication, denoted $\sum$, and $x\mapsto\mset x$ (creation of a singleton multiset) as unit. The \emph{size} of a multiset is the sum of the multiplicities of all its elements. The set of multisets over $X$ of size $k$ is denoted $X^{(k)}$. If $X$ is a measurable space, the measurable sets of $MX$ are those generated by the observations $\len{m\rest B}\ge n$, where $m\in MX$, $B$ is a measurable set of $X$, and $m \rest B$ denotes the restriction of $m$ to $B$: $(m\rest B)(x)=m(x)$ if $x\in B$, $0$ otherwise.

We have included $\infty$ in the general definition of $M$, but it does not play any further role in our development. Thus we will restrict attention to $\naturals^X$, those multisets of $X$ with finite multiplicities. However, we should point out that $\naturals^X$ does not form a monad, except when $X$ is finite.

Let $DX$ be the space of probability measures on a measurable space $X$. For $f:X\to Y$, $Df:DX\to DY$ with $Df(\mu)=\mu\circ f^{-1}$, the \emph{pushforward measure} of $\mu$ under $f$. This is a monad on measurable spaces, often called the \emph{Giry monad} \cite{Giry81}.

We will be particularly interested in $\DNS$, where $\Sigma^*$ is the set of finite-length strings over a finite alphabet $\Sigma$, $\NS$ is the measurable space of multisets of elements of $\Sigma^*$ with finite multiplicities, and $\DNS$ is the space of probability measures over $\NS$. This is the space of behaviors over which our automata and expressions are interpreted. Elements of $\NS$ will be denoted $\alpha,\beta,\ldots$ and elements of $\DNS$ will be denoted $\mu,\nu,\ldots$~.

As a topological space, $\NS$ is homeomorphic to the Baire space $\omega^\omega$, the Cartesian product of $\omega$ copies of $\omega$ with the product topology, where each copy of $\omega$ has the discrete topology. The measurable sets are the Borel sets of this topology. This is a standard Borel space. The Borel sets are generated by the equivalence classes $[\alpha]_n$ of $\equiv_n$, where
\begin{align}
\alpha \equiv_n \beta\ &\Iff\ (\forall x\in\Sigma^*\ \len x\le n\Imp\alpha(x)=\beta(x)) && [\alpha]_n = \set\beta{\beta\equiv_n\alpha}.\label{eq:equivndef}
\end{align}
Every $[\alpha]_n$ has a unique canonical element $\alpha\rest n$ whose support is contained in $\Sigma^{\le n}$, that is, such that $(\alpha\rest n)(x) = 0$ for $\len x>n$, and $\alpha\equiv_n\beta$ iff $\alpha\rest n=\beta\rest n$. The $\equiv_n$-classes $[\alpha]_n$ form a basis for the measurable sets of $\NS$. Every $\alpha\in\NS$ is uniquely determined by its restrictions $\alpha\rest n$; equivalently, $\bigcap_{n\ge 0} [\alpha]_n = \{\alpha\}$.

\subsection{The Distributive Law}

The recently established Beck distributive law $\otimes:MD\to DM$ is a natural transformation whose component $\otimes_X:MDX\to DMX$ for a measurable space $X$ converts a finite multiset of distributions on $X$ to a distribution on finite multisets of $X$. Operationally, one independently samples all the elements in the multiset of distributions to obtain a multiset of elements of $X$; the probability of the sample is the product of the probabilities of its elements.

It must be shown that the distributive law interacts well with the monad structure of $M$ and $D$. These are the Beck conditions. These were verified for finite multisets and finite distributions in \cite{Jacobs21} and for finite multisets and arbitrary distributions over arbitrary measurable spaces in \cite{DashStaton21a,DashStaton21b,Dash23}. In our development, countable multisets with finite multiplicities and continuous distributions do make an appearance in the semantics of automata and expressions. We do not know whether the distributive law holds for countable multisets, but fortunately we need it only for finite multisets, so the results of \cite{Jacobs21,DashStaton21a,DashStaton21b,Dash23} suffice for our purposes.

In this paper, the distributive law is only used in the context $\Mfin\DNS\to D\Mfin(\NS)$ followed by $D\Sigma:D\Mfin(\NS)\to\DNS$, where $\Sigma:\Mfin(\NS)\to\NS$ is multiset union, so the results of \cite{Jacobs21,DashStaton21a,DashStaton21b,Dash23} apply. Additionally, restricting multiset union to domain $\Mfin(\NS)$ ensures that we remain in $\NS$, never generating any infinite multiplicities.

\subsection{Operations on Measures}
\label{sec:opsonmeasures}

The interpretation of the syntactic constructs of our language depends on three semantic operations on measures: product ($\otimes$), probabilistic choice ($\oplus$), and (angelic) nondeterministic choice ($\amp$).
Sequential composition ($\cmp$) is somewhat more involved and is handled separately in \S\ref{sec:composition}.
\begin{itemize}
\item If $\mu$ is a measure on $X$ and $\nu$ is a measure on $Y$, then $\mu\otimes\nu$ is the \emph{product measure} on $X\times Y$ that on a measurable rectangle $A\times B$ gives the value $(\mu\otimes\nu)(A\times B) = \mu(A)\nu(B)$. Operationally, sampling $\mu\otimes\nu$ samples $\mu$ and $\nu$ independently and emits the resulting pair of outcomes. The distributive law $\otimes:MD\to DM$ for finite multisets is a generalized version of this. Operationally, it takes a finite multiset of distributions and samples them all independently, producing a multiset of outcomes.

\item
Let $r\in[0,1]$. If $\mu$ and $\nu$ are measures on the same space $X$, then $\mu\oplus_r\nu = r\mu+(1-r)\nu$ is a measure on $X$. Operationally, sampling $\mu\oplus_r\nu = r\mu+(1-r)\nu$ is equivalent to independently flipping an $r$-biased coin, then sampling $\mu$ on heads or $\nu$ on tails. More generally, let $\set{\mu_n}{n\in I}$ be a finite set of distributions on $X$ and let $\mu$ be a distribution on $I$ such that $n$ occurs with probability $r_n$, $n\in I$. Then $\oplus\mu$ is the flattened measure $\sum_n r_n\mu_n$ on $X$. Operationally, sampling $\oplus\mu$ is equivalent to sampling the distribution $\sum_n r_n \dirac{n}$ to obtain an index $n$, then sampling $\mu_n$ to obtain an element of $X$. The flattening operator $\oplus$ is the multiplication of the Giry monad. 

\item
If $\mu$ and $\nu$ are measures on $\NX$, then $\mu\amp\nu = (\mu\otimes\nu)\circ(+)^{-1}$. Operationally, sampling $\mu\amp\nu$ is equivalent to sampling $\mu$ and $\nu$ independently to obtain two multisets $m,n$ over $X$, then taking their multiset union (pointwise sum) $m+n$. Similarly, there is a generalized version of $\amp$ that applies to larger finite multisets. Like $\otimes$, this takes the form of a natural transformation
\begin{align}
& {\amp}:MDM\to DM && {\amp} = D\Sigma\circ\otimes M.\label{eq:genamp}
\end{align}
Operationally, we independently sample all elements of a finite multiset of distributions over multisets to obtain a finite multiset of multisets, then combine them with multiset union $\Sigma$ (pointwise addition). 
\end{itemize}

\subsection{A Meta-calculus}
\label{sec:metacalculus}

We will sometimes make use of a small meta-calculus for reasoning equationally in an informal operational style. The calculus consists of typing rules and equations involving operators $\sample$ that allows sampling of a distribution and its inverse $\dist$ that constructs a sampleable distribution from a computation. Although not essential for our results, we have nevertheless found it invaluable as a pedagogical aid, as it gives an intuitive operational view of often more obscure denotational arguments.

In addition to the usual rules of the simply typed $\lambda$-calculus, we have constructs $\sample$ for sampling a distribution and $\dist$ for creating a sampleable distribution from a given computation.
\begin{align*}
&
\AxiomC{$\Gamma\vdash d:D\tau$}
\UnaryInfC{$\Gamma\vdash \sample d:\tau$}
\DisplayProof
&&
\AxiomC{$\Gamma\vdash e:\tau$}
\UnaryInfC{$\Gamma\vdash \dist e:D\tau$}
\DisplayProof
\end{align*}
These operators are inverses: 
\begin{align*}
\sample(\dist e) &= e & \dist(\sample d) &= d.
\end{align*}

The meta-calculus is described in more detail in
\S\ref{apx:metacalculus}.

\begin{toappendix}
\label{apx:metacalculus}

Here we provide more details on the meta-calculus described in \S\ref{sec:metacalculus}, along with proofs missing from \S\ref{sec:basics}.

The meta-calculus involves judgments $\Gamma\vdash e:\tau$, meaning that an expression $e$ has type $\tau$ in the typing environment $\Gamma$, where $\Gamma$ gives a typing for the free symbols occurring in $e$. We should think of $e$ as an expression to be evaluated, which might involve sampling from random sources mentioned in $e$ and producing a value of type $\tau$ distributed according to some distribution, depending on the distributions of the sources sampled in $e$. 

In addition to the usual rules of the simply typed $\lambda$-calculus, we have constructs $\sample$ for sampling a distribution and $\dist$ for creating a sampleable distribution from a given computation.
\begin{align*}
&
\AxiomC{$\Gamma\vdash d:D\tau$}
\UnaryInfC{$\Gamma\vdash \sample d:\tau$}
\DisplayProof
&&
\AxiomC{$\Gamma\vdash e:\tau$}
\UnaryInfC{$\Gamma\vdash \dist e:D\tau$}
\DisplayProof
\end{align*}
These operators are inverses: 
\begin{align*}
\sample(\dist e) &= e & \dist(\sample d) &= d.
\end{align*}
To see this, consider possible implementation of these constructs using thunks. We might package an expression $e$ of type $\tau$ in a thunk $\dist e = \lam{()}e$, creating a new sampleable resource of type $D\tau$. To sample it, one would apply it to $()$, thus we should define $\sample d = d\ ()$. Sampling it multiple times is assumed to give independent outcomes. They are inverses, since $\sample(\dist e) = (\lam{()}e)\ ()$, which $\beta$-reduces to $e$, and $\dist(\sample d) = \lam{()}{(d\ ())}$, which $\eta$-reduces to $d$.

We also have the rules
\begin{align*}
&
\AxiomC{$\Gamma\vdash e_1:D\sigma$}
\AxiomC{$\Gamma\vdash e_2:D\tau$}
\BinaryInfC{$\Gamma\vdash e_1\otimes e_2:D(\sigma\times\tau)$}
\DisplayProof
&
e_1\otimes e_2 &= \dist(\sample e_1,\sample e_2)
\end{align*}
for forming a product distribution, and
\begin{align*}
&
\AxiomC{$\Gamma\vdash e:D(\sigma_1\times\sigma_2)$}
\UnaryInfC{$\Gamma\vdash D\pi_i(e):D\sigma_i$}
\DisplayProof
&
D\pi_i(e) &= \dist(\pi_i(\sample e))
\end{align*}
for taking marginals. For nondeterministic choice, we have
\begin{align}
&
\AxiomC{$\Gamma\vdash e_1:DM\tau$}
\AxiomC{$\Gamma\vdash e_2:DM\tau$}
\BinaryInfC{$\Gamma\vdash e_1\amp e_2:DM\tau$}
\DisplayProof
& e_1\amp e_2 = \dist(\sample e_1 + \sample e_2)\label{eq:nondettype}
\end{align}
where $+$ denotes multiset union. For multisets, we have
\begin{align*}
&
\AxiomC{$\Gamma\vdash f:\sigma\to\tau$}
\UnaryInfC{$\Gamma\vdash Mf:M\sigma\to M\tau$}
\DisplayProof
&&
\AxiomC{$\Gamma\vdash e_1:M\tau$}
\AxiomC{$\Gamma\vdash e_2:M\tau$}
\BinaryInfC{$\Gamma\vdash e_1+e_2:M\tau$}
\DisplayProof
\end{align*}
Also, pushforward measures can be expressed:
\begin{align*}
&
\AxiomC{$\Gamma\vdash e:D\sigma$}
\AxiomC{$\Gamma\vdash f:\sigma\to\tau$}
\BinaryInfC{$\Gamma\vdash Df(e):D\tau$}
\DisplayProof
&
Df(e) &= \dist(f(\sample e))
\end{align*}

Here is an example of a proof that derives the typing rule \eqref{eq:nondettype}.
\begin{align*}
&
\AxiomC{$\Gamma\vdash e_1:DM\tau$}
\UnaryInfC{$\Gamma\vdash \sample e_1:M\tau$}
\AxiomC{$\Gamma\vdash e_2:DM\tau$}
\UnaryInfC{$\Gamma\vdash \sample e_2:M\tau$}
\BinaryInfC{$\Gamma\vdash\sample e_1+\sample e_2:M\tau$}
\UnaryInfC{$\Gamma\vdash \dist(\sample e_1+\sample e_2):DM\tau$}
\UnaryInfC{$\Gamma\vdash e_1\amp e_2:DM\tau$}
\DisplayProof
\end{align*}

\end{toappendix}

\subsection{Injective Monoid Actions}
\label{sec:injectivemonoidactions}

Another important concept in our semantics is \emph{injective monoid actions}. This concept will be crucial in the semantics of sequential composition, which is unlike other operations in that its effects are nonlocal. The significance of the basic constructs here will only become clear later in \S\ref{sec:semanticsofexpressions}, so we suggest skipping this section and the following one on first reading.

Suppose $S$ is a monoid acting on a set $X$. Thus for $s,t\in S$, we have $1\cdot x = x$ and $st\cdot x = s\cdot(t\cdot x)$. Suppose further that the monoid action is injective; that is, if $s\cdot x=s\cdot y$, then $x=y$.

There is a canonical way to lift the monoid action to multisets in $\NX$, namely
\begin{align*}
(s\cdot m)(y) &= \begin{cases}
m(x) & \text{if $y = s\cdot x$}\\ 
0 & \text{if $y\ne s\cdot x$ for any $x$.}
\end{cases}
\end{align*}
This is well defined, as the action on $X$ is injective. Moreover, the lifted action on $\NX$ is injective: if $s\cdot m=s\cdot n$, then for all $x$, $m(x) = (s\cdot m)(s\cdot x) = (s\cdot n)(s\cdot x) = n(x)$, so $m=n$.

There is a canonical way to lift the monoid action to $2^X$, namely $s\cdot A = \set{s\cdot x}{x\in A}$. This is actually a special case of the lifted action on $\NX$ described above.

Finally, there is a canonical way to lift the monoid action to $DX$, provided $s\cdot-:X\to X$ is a measurable function. We take $s\cdot\mu = \mu\circ(s\cdot-)^{-1}$, the pushforward measure of $\mu$ under $s\cdot-:X\to X$. Moreover, the lifted action on $DX$ is injective: for all measurable $A$, we have $(s\cdot-)^{-1}(s\cdot A) = A$ since $s\cdot-$ is injective on $X$, so if $s\cdot\mu=s\cdot\nu$, then
\begin{align*}
\mu(A) &= \mu((s\cdot-)^{-1}(s\cdot A)) = (s\cdot\mu)(s\cdot A) = (s\cdot\nu)(s\cdot A)
= \nu((s\cdot-)^{-1}(s\cdot A)) = \nu(A),
\end{align*}
therefore $\mu=\nu$.

In our application, the monoid $S$ will be the free monoid $\Sigma^*$. By the above arguments, $\Sigma^*$ acts injectively on $\Sigma^*$, $\NS$, $\powerset{\NS}$, and $\DNS$: for $x,y\in\Sigma^*$, $m\in \NS$, $A\subs \NS$, and $\mu\in \DNS$,
\begin{align*}
x\cdot y &= xy
&
(x\cdot m)(z) &= \begin{cases}
m(y), & \text{if $z=xy$},\\
0, & \text{if $x$ is not a prefix of $z$}
\end{cases}\\[1ex]
x\cdot A &= \set{x\cdot m}{m\in A}
&
x\cdot\mu &= \mu\circ(x\cdot-)^{-1}.
\end{align*}

\subsection{Composition}
\label{sec:composition}

For $\nu\in\DNS$, the operation $-\cdot\nu : \Sigma^* \to \DNS$ introduced in \S\ref{sec:injectivemonoidactions} gives rise to a bind operation
\begin{align*}
& {\bind} : \DNS \times (\Sigma^* \to \DNS) \to \DNS
& \mu\bind-\cdot\nu : \DNS
\end{align*}
that will be used in the definition of sequential composition. We only ever apply $\bind$ in the form $\mu\bind-\cdot\nu$, which ensures that we never generate any infinite multiplicities, as explained below.

The bind operation $\mu\bind-\cdot\nu$ is defined by first extending $-\cdot\nu:\Sigma^*\to\DNS$ to domain $\NS$, then integrating with respect to $\mu$. The extension is also denoted $-\cdot\nu:\NS\to\DNS$ and defined by
\begin{align*}
\beta\cdot\nu = {\amp}M(-\cdot\nu)(\beta) = \otimes(M(-\cdot\nu)(\beta))\circ\Sigma^{-1}.
\end{align*}
Operationally, to sample $\beta\cdot\nu$, we independently sample $x\cdot\nu$ for each $x\in\beta$, then take the multiset union (pointwise sum) of the outcomes. In the notation of the meta-calculus of \S\ref{sec:metacalculus},
\begin{align*}
\sample(\beta\cdot\nu) &= \sum_{x\in\beta}\ \sample(x\cdot \nu).
\end{align*}
Although this is formally an infinite sum, note that only $x\in\beta$ with $x$ a prefix of $y$ can contribute nonzero multiplicity to $y$ in the final result, as all multisets produced by $x\cdot\nu$ contain only strings that have $x$ as a prefix, and there are only finitely many occurrences of such $x$ in $\beta$. Thus the final outcome is a multiset with finite multiplicities.

Finally, we integrate with respect to $\mu$ by Lebesgue integration to get the measure
\begin{align}
(\mu\bind -\cdot\nu)(A) = \int_{\beta\in\NS}\ (\beta\cdot\nu)(A)\,d\mu.\label{eq:binddef}
\end{align}
Operationally, to sample $\mu\bind -\cdot\nu$, we sample $\mu$ to obtain a multiset $\beta$, then sample $x\cdot\nu$ for each $x\in\beta$ and take their multiset union. In the notation of the meta-calculus of \S\ref{sec:metacalculus},
\begin{align*}
\sample(\mu\bind -\cdot\nu) &= \letin\beta{\sample\mu}{\sum_{x\in\beta}\ \sample(x\cdot\nu)}.
\end{align*}

In order to integrate, we need $-\cdot\nu:\NS\to\DNS$ to be a measurable function on $\NS$. This is established in the following lemmas, along with some other useful properties.
\begin{lemmarep}
\label{lem:injproperties}
\begin{enumerate}[{\upshape(i)}]
\item
If $\mu_i \equiv_{n} \nu_i$ for all $i\in I$, then $\sum_{i\in I} \mu_i \equiv_{n} \sum_{i\in I} \nu_i$.
\item
If $\nu_1\equiv_n\nu_2$, then $x\cdot\nu_1 \equiv_{n+\len x}x\cdot\nu_2$.
\item
If $\beta_1\equiv_n\beta_2$, then $\beta_1\cdot\nu\equiv_n\beta_2\cdot\nu$.
\end{enumerate}
\end{lemmarep}
\begin{proof}
For (i), if $\mu_i \equiv_n \nu_i$ for all $i\in I$, then for all $i\in I$ and $\alpha\in\NS$, $\mu_i([\alpha]_n) = \nu_i([\alpha]_n)$. Then for all $\alpha\in\NS$,
\begin{align*}
(\sum_{i\in I} \mu_i)([\alpha]_n)
&= \sum_{i\in I} \mu_i([\alpha]_n)
= \sum_{i\in I} \nu_i([\alpha]_n)
= (\sum_{i\in I} \nu_i)([\alpha]_n),
\end{align*}
which says that $\sum_{i\in I} \mu_i \equiv_n \sum_{i\in I} \nu_i$.

For (ii), it follows that $(\beta\cdot\nu)([\alpha]_n) = (\beta\rest n\cdot\nu)([\alpha]_n)$, where $\beta\rest n$ is the unique finite multiset whose support is contained in $\Sigma^{\le n}$ and that agrees with $\beta$ on that set. Now if $\beta_1\equiv_n\beta_2$, then $\beta_1\rest n=\beta_2\rest n$, so
\begin{align*}
(\beta_1\cdot\nu)([\alpha]_n) &= (\beta_1\rest n\cdot\nu)([\alpha]_n) = (\beta_2\rest n\cdot\nu)([\alpha]_n) =(\beta_2\cdot\nu)([\alpha]_n),
\end{align*}
therefore $\beta_1\cdot\nu\equiv_n\beta_2\cdot\nu$.

For (iii), if $\nu_1\equiv_n\nu_2$, then
\begin{align*}
(x\cdot\nu_1)([\gamma]_{n+\len x})
&= \begin{cases}
\nu_1([\alpha]_n) & \text{if $\gamma = x\cdot\alpha$}\\
0 & \text{if $\gamma \ne x\cdot\alpha$ for any $\alpha$}
\end{cases}\\
&= \begin{cases}
\nu_2([\alpha]_n) & \text{if $\gamma = x\cdot\alpha$}\\
0 & \text{if $\gamma \ne x\cdot\alpha$ for any $\alpha$}
\end{cases}
\qquad = (x\cdot\nu_2)([\gamma]_{n+\len x}),
\end{align*}
so $x\cdot\nu_1 \equiv_{n+\len x}x\cdot\nu_2$.
\end{proof}

\begin{lemmarep}\ 
\label{lem:bind}
\begin{enumerate}[{\upshape(i)}]
\item
If $\beta\equiv_m 0$ and $\nu_1\equiv_n\nu_2$, then $\beta\cdot\nu_1 \equiv_{m+n+1} \beta\cdot\nu_2$.
\item
If $\mu([0]_m)=1$ and $\nu_1\equiv_n\nu_2$, then $\mu\bind-\cdot\nu_1\equiv_{m+n+1}\mu\bind-\cdot\nu_2$.
\end{enumerate}
\end{lemmarep}
\begin{proof}
For (i), since $\beta\equiv_m 0$, we have $\len x\ge m+1$ for all $x\in\supp\beta$. By Lemma \ref{lem:injproperties}(ii), since $\nu_1\equiv_n\nu_2$, we have $x\cdot\nu_1 \equiv_{n+\len x} x\cdot\nu_2$ for all $x\in\supp\beta$, therefore $x\cdot\nu_1 \equiv_{n+m+1} x\cdot\nu_2$ for all $x\in\supp\beta$. By Lemma \ref{lem:injproperties}(i), $\beta\cdot\nu_1 = \sum_{x\in\beta} x\cdot\nu_1 \equiv_{n+m+1} \sum_{x\in\beta} x\cdot\nu_2 = \beta\cdot\nu_2$.

For (ii), if $\mu([0]_m) = 1$, then
\begin{align*}
\mu(\set\beta{\beta\not\equiv_m 0}) &= \mu(\bigcup\set{[\beta]_m}{[\beta]_m\ne[0]_m})
= \sum\set{\mu([\beta]_m)}{[\beta]_m\ne[0]_m})
= 0,
\end{align*}
so for any $\nu$,
\begin{align*}
(\mu\bind-\cdot\nu)(A)
&= \int_{\beta\in\NS}(\beta\cdot\nu)(A)\,d\mu
= \int_{\substack{\beta\in\NS\\\beta\equiv_m 0}}(\beta\cdot\nu)(A)\,d\mu.
\end{align*}
By (i), if $\beta\equiv_m 0$, then $\beta\cdot\nu_1 \equiv_{m+n+1} \beta\cdot\nu_2$. Then
\begin{align*}
(\mu\bind-\cdot\nu_1)([\alpha]_{m+n+1})
&= \int_{\substack{\beta\in\NS\\\beta\equiv_m 0}}(\beta\cdot\nu_1)([\alpha]_{m+n+1})\,\mu([\beta]_{m+n+1})\\
&= \int_{\substack{\beta\in\NS\\\beta\equiv_m 0}}(\beta\cdot\nu_2)([\alpha]_{m+n+1})\,\mu([\beta]_{m+n+1})\\
&= (\mu\bind-\cdot\nu_2)([\alpha]_{m+n+1}).
\tag*\qedhere
\end{align*}
\end{proof}

\begin{lemmarep}
\label{lem:sumprod}
\begin{align*}
(a\cdot-)^{-1}([\alpha]_n)
&= \begin{cases}
\NS, & \text{if $n=0$ and $\alpha(\eps)=0$},\\
[\beta]_{n-1}, & \text{if $n\ge 1$ and $\alpha\equiv_n a\cdot \beta$},\\
\emptyset, & \text{otherwise.}
\end{cases}
\end{align*}
\end{lemmarep}
\begin{proof}
For $n=0$,
\begin{align*}
(p\cdot-)^{-1}([\alpha]_0)
&= \set{\beta}{p\cdot \beta\equiv_0 \alpha}
= \set{\beta}{\alpha(\eps)=0}
= \begin{cases}
\NS, & \text{if $\alpha(\eps)=0$},\\
\emptyset, & \text{otherwise.}
\end{cases}
\end{align*}
For $n\ge 1$, if $a\cdot \beta\equiv_n \alpha$ then $a\cdot \beta\equiv_{n-1} \alpha$, and if $\gamma\equiv_{n-1} \beta$ then $a\cdot \gamma\equiv_{n} a\cdot \beta$. It follows that
\begin{align*}
a\cdot \beta\equiv_n \alpha\ &\Iff\ \exists \gamma\ \ \alpha \equiv_n a\cdot \gamma\wedge \beta\equiv_{n-1} \gamma,
\end{align*}
so
\begin{align*}
(a\cdot-)^{-1}([\alpha]_n)
&= \set{\beta}{a\cdot \beta\equiv_n \alpha}
= \set{\beta}{\exists \gamma\ \ \alpha \equiv_n a \cdot \gamma \wedge \beta\equiv_{n-1} \gamma}\\
&= \begin{cases}
[\gamma]_{n-1}, & \text{if $\alpha\equiv_n a \cdot \gamma$,}\\
\emptyset, & \text{otherwise.}
\end{cases}
\tag*\qedhere
\end{align*}
\end{proof}

\begin{lemmarep}
For $\nu\in\DNS$, the map $-\cdot\nu:\NS\to\DNS$ is a measurable function.
\end{lemmarep}
\begin{proof}
Since $\beta\equiv_n\beta\rest n$, by Lemma \ref{lem:injproperties}(iii) we have $\beta\cdot\nu\equiv_n\beta\rest n\cdot\nu$. Then the preimage of a basic measurable set is
\begin{align*}
\set{\beta}{(\beta\cdot\nu)([\alpha]_n) \ge r}
&= \set{\beta}{(\beta\rest n\cdot\nu)([\alpha]_n) \ge r}
= \bigcup\,\set{[\beta]_n}{\beta=\beta\rest n,\ (\beta\cdot\nu)([\alpha]_n) \ge r},
\end{align*}
a countable union of $\equiv_n$-classes, thus a measurable set.
\end{proof}

\section{Automata and Expressions}
\label{sec:automataandexpressions}

In this section we introduce a new class of automata and expressions with probabilistic choice and angelic nondeterminism. We will give the formal definition of the automata and expressions and their semantics, along with some examples.

\subsection{Automata}
\label{sec:automata}

An automaton consists of a finite set of states and transitions of four types, along with a designated start state. The four types are
 \begin{itemize}
 \item
 nondeterministic choice states labeled $\amp$ with a finite multiset of successor states;
 \item
 probabilistic states labeled $\oplus$ with a finite distribution on successor states;
 \item
 terminal states labeled $\skp$ (accept) or $\fail$ (reject) with no successor states; and
 \item
 action states labeled $a\in\Sigma$ with one successor state.
 \end{itemize}
 
Informally, an automaton can be viewed either as an \emph{acceptor} that takes a string in $\Sigma^*$ as input or as an \emph{enumerator} that generates strings in $\Sigma^*$. In either view, the operation of the automaton can be described in terms of agents. At any time in the computation, there can be multiple agents, each occupying a state and acting independently of the other agents.

In the enumeration view, a single agent originates at the start state. Thereafter, the computation proceeds as follows:
\begin{itemize}
\item
Each agent visiting a nondeterministic state $s$ labeled $\amp$ replicates itself $n-1$ times, where $n$ is the size of the multiset of successors of $s$. The original agent and its copies are distributed to the successor multiset, respecting multiplicities. That is, if $t$ occurs in the successor multiset with multiplicity $k$, then $k$ new agents are created at $t$.
\item
Each agent visiting a probabilistic state $s$ labeled $\oplus$ independently samples the distribution associated with $s$, yielding a successor state to which the agent moves.
\item
An agent visiting a terminal state labeled $\skp$ outputs (enumerates) the string of letters of $\Sigma$ it has seen so far since the start, where \emph{seen} means having visited an action state with that label. An agent visiting a terminal state labeled $\fail$ rejects.
\item
An agent visiting an action state labeled $a$ with successor $t$ appends $a$ to the string of letters it has seen and moves to $t$.
\end{itemize}
The behavior for acceptors is largely the same, except that at $\skp$ states, the agent accepts if the entire input string has been scanned and rejects otherwise, and at action state $a$ with successor $t$, the agent advances past the symbol $a$ in the input string and moves to $t$ if the next input letter is $a$, otherwise rejects.

The string $x$ is generated (or accepted) with multiplicity $k$ if $k$ is the number of agents generating (or accepting) that string. Of course, this occurs with some probability, depending on the probabilistic choices of the agents. Moreover, the probabilities for different strings and different multiplicities may be correlated. Thus the behavior of the automaton is best described by a joint distribution on the space of multisets of strings $\DNS$.

To avoid infinite multiplicities, we impose the restriction that every cycle in the automaton must contain an action state. This is known as the \emph{productivity assumption}\footnote{The productivity assumption can be weakened to allow cycles containing probabilistic states only. However, this comes at some cost in the complexity of the presentation, so we do not pursue this option here.}. It is crucial for the coalgebraic treatment of automata and expressions.

Formally, an automaton is a tuple $(S,\Sigma,\ell,\del)$, where $S$ is a finite set of states, $\Sigma$ is a finite alphabet of input letters, $\ell:S\to\{{\amp},\oplus,\skp,\fail\}\cup\Sigma$ is a labeling function, and $\del:S\to DS + \Ns + S + 1$ are the transitions, such that:
\begin{itemize}
\item
if $\ell(s)=\oplus$, then $\del(s)\in DS$, that is, $\del(s)$ is a probability measure on $S$;
\item
if $\ell(s)={\amp}$, then $\del(s)\in\Ns$, that is, $\del(s)$ is a multiset of elements of $S$;
\item
if $\ell(s)\in\Sigma$, then $\del(s)\in S$, that is, $s$ has one successor; and
\item
if $\ell(s)\in\{\skp,\fail\}$, then $s$ has no successors.
\end{itemize}
We might also wish to designate a particular start state $s_0\in S$.
States $s$ with $\ell(s)=\oplus$ are called \emph{probabilistic states}, 
those with $\ell(s)={\amp}$ are called \emph{choice states},
those with $\ell(s)\in\Sigma$ are called \emph{action states},
and those with $\ell(s)\in\{\skp,\fail\}$ are called \emph{terminal states}.

The productivity assumption has the following consequence: For every state $s$ and $k\ge 0$, every path starting from $s$ of length at least $k\len S$ visits at least $k$ action states.

Fig.~\ref{fig:coalg} is an illustrative example for a two-letter alphabet $\Sigma=\{a,b\}$ with $S=\{s,t,u,v,\dots\}$.
\begin{figure}
\centering
\begin{tikzpicture}[->, >=stealth', node distance=2cm, auto, bend angle=35, scale=0.8]
\small
 \node at (-.3,.3) {$s$};
 \node (top) at (0,0) {$\oplus$};
 \node (amp1) at (-6,-2) {$\&$};
 \path (top) edge node[swap] {$p$} (amp1);
 \node (amp2) at (-2,-2) {$\&$};
 \path (top) edge node[swap, pos=.8] {$q$} (amp2);
 \node (amp3) at (2,-2) {$\&$};
 \path (top) edge node[pos=.8] {$r$} (amp3);
 \node (amp4) at (6,-2) {$\&$};
 \path (top) edge node {$1-(p+q+r)$} (amp4);
 \node (amp1a) at (-7.3,-3.6) {$\skp$};
 \path (amp1) edge (amp1a);
 \node (amp1b) at (-6.4,-3.6) {$\skp$};
 \path (amp1) edge (amp1b);
 \node (amp1c) at (-5.6,-3.6) {$a$};
 \path (amp1) edge (amp1c);
 \node (amp1d) at (-4.8,-3.6) {$b$};
 \path (amp1) edge (amp1d);
 \node (amp1cx) at (-5.6,-4.6) {$s$};
 \path (amp1c) edge (amp1cx);
 \node (amp1dx) at (-4.8,-4.6) {$t$};
 \path (amp1d) edge (amp1dx);
 \node (amp2a) at (-2.8,-3.6) {$\skp$};
 \path (amp2) edge (amp2a);
 \node (amp2b) at (-2,-3.6) {$a$};
 \path (amp2) edge (amp2b);
 \node (amp2c) at (-1.2,-3.6) {$a$};
 \path (amp2) edge (amp2c);
 \node (amp2bx) at (-2,-4.6) {$u$};
 \path (amp2b) edge (amp2bx);
 \node (amp2cx) at (-1.2,-4.6) {$v$};
 \path (amp2c) edge (amp2cx);
 \node (amp3a) at (.3,-3.6) {$\skp$};
 \path (amp3) edge (amp3a);
 \node (amp3b) at (1.2,-3.6) {$\skp$};
 \path (amp3) edge (amp3b);
 \node (amp3c) at (2,-3.6) {$a$};
 \path (amp3) edge (amp3c);
 \node (amp3d) at (2.8,-3.6) {$b$};
 \path (amp3) edge (amp3d);
 \node (amp3e) at (3.6,-3.6) {$b$};
 \path (amp3) edge (amp3e);
 \node (amp3cx) at (2,-4.6) {$s$};
 \path (amp3c) edge (amp3cx);
 \node (amp3dx) at (2.8,-4.6) {$t$};
 \path (amp3d) edge (amp3dx);
 \node (amp3ex) at (3.6,-4.6) {$t$};
 \path (amp3e) edge (amp3ex);
 \node (amp4a) at (4.8,-3.6) {$a$};
 \path (amp4) edge (amp4a);
 \node (amp4b) at (5.6,-3.6) {$a$};
 \path (amp4) edge (amp4b);
 \node (amp4c) at (6.4,-3.6) {$a$};
 \path (amp4) edge (amp4c);
 \node (amp4d) at (7.2,-3.6) {$b$};
 \path (amp4) edge (amp4d);
 \node (amp4ax) at (4.8,-4.6) {$s$};
 \path (amp4a) edge (amp4ax);
 \node (amp4bx) at (5.6,-4.6) {$s$};
 \path (amp4b) edge (amp4bx);
 \node (amp4cx) at (6.4,-4.6) {$t$};
 \path (amp4c) edge (amp4cx);
 \node (amp4dx) at (7.2,-4.6) {$t$};
 \path (amp4d) edge (amp4dx);
\end{tikzpicture}
\caption{Fragment of an automaton}
\label{fig:coalg}
\end{figure}
In this example, starting from state $s$,
\begin{itemize}
\item
with probability $p$, $\eps$ is accepted with multiplicity 2, and the automaton transitions to $\mset s$ on input $a$ and to $\mset t$ on input $b$;
\item
with probability $q$, $\eps$ is accepted with multiplicity 1, and the automaton transitions to $\mset{u,v}$ on input $a$ and to $\emptyset$ on input $b$;
\item
with probability $r$, $\eps$ is accepted with multiplicity 2, and the automaton transitions to $\mset s$ on input $a$ and to $\mset{t,t}$ on input $b$; and
\item
with probability $1-(p+q+r)$, $\eps$ is accepted with multiplicity 0 (i.e., not accepted), and the automaton transitions to $\mset{s,s,t}$ on input $a$ and to $\mset t$ on input $b$.
\end{itemize}
From state $s$, the probability that $\eps$ is accepted with multiplicity $0$, $1$, or $2$ is $1-(p+q+r)$, $q$, and $p+r$, respectively. On input $a$, the automaton transitions to $\mset s$, $\mset{u,v}$, or $\mset{s,s,t}$ with probabilities $p+r$, $q$, and $1-(p+q+r)$, respectively.

Further examples will be given below in \S\ref{sec:examples}.

\subsection{Semantics of Automata}
\label{sec:semanticsofautomata}

Every state in $S$ represents a distribution over multisets of strings in $\DNS$. Let $\dirac m$ denote the Dirac (point mass) measure on the multiset $m$.
The semantic map $\sem-:S\to\DNS$ is defined coinductively:
\begin{align*}
\sem s &= \begin{cases} 
\dirac{\mset\eps}, & \text{if $\ell(s)=\skp$}\\
\dirac{\mset{}}, & \text{if $\ell(s)=\fail$}\\
a \cdot\sem t, & \text{if $\ell(s)= a \in\Sigma$ and $\del(s)=t$}\\
\sum_i r_i\sem{t_i}, & \text{if $\ell(s)=\oplus$ and $\del(s)=\sum_i r_it_i$}\\
{\amp}(M\sem-(m)), & \text{if $\ell(s)={\amp}$ and $\del(s)=m$.}
\end{cases}
\end{align*}
In the last case, $M\sem-$ is the map that, given a multiset $m$, applies $\sem-$ to every element of $m$ and takes the multiset of results, and $\amp$ is the semantic operation described in \S\ref{sec:opsonmeasures}. Operationally, interpret all elements of $m$ by $\sem-$ to obtain a multiset of distributions in $\DNS$, then sample all of them independently and take the multiset union of the outcomes.

For the special case of binary $\amp$ and $\oplus_r$, these definitions reduce to
\begin{align*}
\sem{s\amp t} &= (\sem s\otimes\sem t)\circ{(+)}^{-1}
& \sem{s\oplus_r t} &= r\sem s + (1-r)\sem t.
\end{align*}

We will argue in Lemma \ref{thm:autsemwelldefined} that the map $\sem-$ is well defined. Briefly, the productivity assumption ensures that a map modeling the coinductive definition is contractive in a certain complete metric space, thus by the Banach fixpoint theorem has a unique fixpoint $\sem-$.

\subsection{Expressions}
\label{sec:expressions}

Let $\Var=\{x,y,\ldots\}$ be a set of variables and $\Sigma = \{a,b,\ldots\}$ a finite set of letters disjoint from $\Var$. The language of expressions is given by the BNF grammar
\begin{align*}
e\ &::=\ x \mid a \mid e_1 \amp e_2 \mid e_1 \oplus_r e_2 \mid e_1\cmp e_2 \mid \fix xe \mid \skp \mid \fail
\end{align*}
The operator $\cmp$ is for sequential composition. In any expression of the form $e_1\cmp e_2$, we require that the left operand $e_1$ be closed (contain no free variables), a restriction which characterizes linear recursion. Without it, we could have for example $\fix x{\skp \amp a;x;b}$, corresponding to the non-regular context-free language $\set{a^nb^n}{n\ge 0}$.
Thus expressions are the appropriate analog of regular (rational) expressions in this context, though the fixpoint operator $\fix xe$ is somewhat more expressive than the usual Kleene star. Expressions with star correspond to fixpoint expressions with one occurrence of the variable in the body, which is apparently too weak to encode automata; at least we do not see how. We need the more general fixpoint form allowing multiple occurrences of the variable in order to obtain a Kleene theorem. The general form we use is motivated by the fixpoint expressions of the modal $\mu$-calculus.

We also require that in fixpoint expressions $\fix xe$, all paths from the root of $e$ to a free occurrence of $x$ must pass through some $a\in\Sigma$. This is another manifestation of the \emph{productivity assumption}.

\subsection{Semantics of Expressions}
\label{sec:semanticsofexpressions}

Like automata, expressions $e$ are interpreted as measures on multisets of strings $\semr e:\DNS$ relative to an environment $\rho:\Var\to\DNS$. The environment $\rho$ is used to interpret free variables. A closed expression (one in which all variables are bound by $\mathsf{fix}$) does not need $\rho$. The notation $\dirac m$ denotes the Dirac (point mass) measure on the multiset $m$. The definitions for $e_1 \cmp e_2$ and $\fix xe$ contain some undefined notation, which we will explain below.
\begin{gather*}
\semr\skp = \dirac{\mset\eps} \qquad\quad
\semr\fail = \dirac{\mset{}} \qquad\quad
\semr a = \dirac{\mset a} \qquad\quad
\semr x = \rho(x)\\
\semr{e_1 \oplus_r e_2} = r\semr{e_1} + (1-r)\semr{e_2} \qquad
\semr{e_1 \amp e_2} = \semr{e_1}\amp\semr{e_2} = (\semr{e_1}\otimes\semr{e_2})\circ(+)^{-1}\\
\semr{e_1 \cmp e_2} = \sem{e_1} \bind -\cdot\semr{e_2} \qquad\qquad
\semr{\fix xe} = \semr e[\semr{\fix xe}/x]
\end{gather*}

The semantics of sequential composition $\cmp$ is based on Kleisli composition involving the bind operation ($\bind$)
introduced in \S\ref{sec:composition}. 
Note also that 
$\rho$ is not needed for $\sem{e_1}$ because of the restriction that $e_1$ must be closed.
Operationally, to sample $\semr{e_1\cmp e_2}$, we sample $\sem{e_1}$ to obtain a multiset $\beta$, then sample $x\cdot\semr{e_2} = \semr{x\cmp e_2}$ for each $x\in\beta$ and take their multiset union (pointwise sum). In the notation of the meta-calculus of \S\ref{sec:metacalculus},
\begin{align*}
\sample\semr{e_1\cmp e_2} &= \letin\beta{\sample\sem{e_1}}{\sum_{x\in\beta}\ \sample\semr{x\cmp e_2}}.
\end{align*}
We have required that the left operand $e_1$ be closed in compositions $e_1\cmp e_2$. Absent this restriction, besides enabling nonlinear behavior, in the expression  in $\sem{e_1}\rho\bind-\cdot\sem{e_2}\rho$, the parts of $\sem{e_1}\rho$ supplied by $\rho$ to the free variables of $e_1$ would feed into $\sem{e_2}\rho$, which would break compositionality.

In the definition of $\semr{\fix xe}$, the notation $\rho\subst\mu x$ refers to the environment $\rho$ with $x$ rebound to $\mu$. It appears that the definition is circular. However, we will show in Theorem \ref{thm:expsemwelldefined} that the semantics is well defined due to the productivity assumption.

The expression $\fix xe$ represents the unique solution of the equation $x=e$
in $\DNS$, where $e$ may contain free occurrences of $x$. A special case is the traditional $^*$ operator of Kleene algebra
\begin{align*}
e^* &= \fix x{\skp\amp(e\cmp x)},
\end{align*}
which is the unique solution of the equation $x = \skp\amp(e\cmp x)$
in $\DNS$. Thus we will have
\begin{align*}
\fix xe &= e[(\fix xe)/x] & e^* &= \skp\amp (e\cmp e^*)
\end{align*}
as shown formally in Lemma \ref{lem:substrebind} below.

\subsection{Examples}
\label{sec:examples}

Fig.~\ref{fig:exampleexpressions} shows some small examples of expressions and their equivalent automata.

\begin{figure}
  \centering
  \begin{subfigure}[b]{0.3\textwidth}
    \centering
    \begin{align*}
    \begin{tikzpicture}[->, >=stealth', auto, node distance=16mm]
    \useasboundingbox (0,-1.7) rectangle (0,.8);
    \small
    \node (A) {$\amp$};
    \node (Al) [left of=A, node distance=3mm, yshift=1mm] {$s$};
    \node (B) [below left of=A] {$\skp$};
    \node (C) [below right of=A] {$a$};
    \path (A) edge (B);
    \path (A) edge (C);
    \path (C) edge[bend right] (A);
    \end{tikzpicture}
    \end{align*}
    \caption{$a^*$}
    \label{fig:astar}
  \end{subfigure}%
  ~ 
  \begin{subfigure}[b]{0.3\textwidth}
    \centering
    \begin{align*}
    \begin{tikzpicture}[->, >=stealth', auto, node distance=16mm]
    \useasboundingbox (0,-2.5) rectangle (0,0);
    \small
    \node (A) {$\amp$};
    \node (Al) [left of=A, node distance=3mm, yshift=1mm] {$s$};
    \node (B) [below left of=A] {$\skp$};
    \node (C) [below right of=A] {$a$};
    \node (Cx) [below of=C, node distance=10mm] {$\amp$};
    \node (Cl) [left of=Cx, node distance=3mm, yshift=1mm] {$t$};
    \node (D) [below left of=Cx, node distance=12mm] {};
    \node (E) [below right of=Cx, node distance=10mm] {$a$};
    \path (A) edge (B);
    \path (A) edge (C);
    \draw (Cx) .. controls (D) .. (A);
    \path (C) edge (Cx);
    \path (Cx) edge (E);
    \path (E) edge[bend right] (Cx);
    \end{tikzpicture}
    \end{align*}
    \caption{$(aa^*)^*$}
    \label{fig:aastarstar}
  \end{subfigure}
  ~ 
  \begin{subfigure}[b]{0.3\textwidth}
    \centering
    \begin{align*}
    \begin{tikzpicture}[->, >=stealth', auto, node distance=16mm]
    \useasboundingbox (0,-2) rectangle (0,.5);
    \small
    \node (A) {$\amp$};
    \node (Al) [left of=A, node distance=3mm, yshift=1mm] {$s$};
    \node (B) [below left of=A] {$\skp$};
    \node (C) [below right of=A] {$\oplus$};
    \node (Cl) [left of=C, node distance=3mm] {$t$};
    \node (D) [below left of=C, node distance=10mm] {$a$};
    \node (E) [below right of=C, node distance=10mm] {$b$};
    \path (A) edge (B);
    \path (A) edge (C);
    \path (C) edge (D);
    \path (C) edge (E);
    \path (E) edge[bend right] (A);
    \path (D) edge[bend angle=10, bend left] (A);
    \end{tikzpicture}
    \end{align*}
    \caption{$(a \oplus_{1/2} b)^*$}
    \label{fig:aobstar}
  \end{subfigure}
  \caption{Examples of automata}
  \label{fig:exampleexpressions}
\end{figure}

The automaton of Fig.~\ref{fig:exampleexpressions}(a) with start state $s$ corresponds to the expression $a^*$. The behavior $\sem s = \sem{a^*}$ is a point mass on the multiset $a^n\mapsto 1$. This can be seen by solving a recurrence. Let $g(n)$ be the number of paths on which $a^n$ is accepted, starting at $s$. Then $g(0)=1$ and $g(n+1)=g(n)$, so $g(n)=1$ for all $n\ge 0$. There is no probabilistic choice in this example, which is why the outcome is a point mass.

We cannot have $a^{**}$ by the productivity assumption, but $(aa^*)^*$ is allowed. This corresponds to the automaton of Fig.~\ref{fig:exampleexpressions}(b) with start state $s$. 
The behavior $\sem s= \sem{(aa^*)^*}$ is a point mass on the multiset $a^0\mapsto 1$, $a^{n+1}\mapsto 2^n$. This can be seen by solving recurrences for $f(n)$ and $g(n)$, the number of paths accepting $a^n$ starting from states $s$ and $t$, respectively:
\begin{align*}
& f(0) = 1 && g(0) = 1 && f(n + 1) = g(n) && g(n + 1) = f(n + 1) + g(n)
\end{align*}
so $g(n + 1) = 2g(n)$, giving $g(n) = 2^n$ and $f(n + 1) = 2^n$.
This example shows that multiplicities, though always finite, can grow exponentially.

The automaton of Fig.~\ref{fig:exampleexpressions}(c) with start state $s$ corresponds to the expression $(a\oplus_{1/2}b)^*$. The behavior $\sem s=\sem{(a\oplus_{1/2} b)^*}$ is the uniform distribution over all \emph{maximal} multisets of $\{a,b\}^*$ with multiplicities at most 1 (that is, they are sets) and linearly ordered by the prefix relation. For example, $\mset{\eps,a,aa,aab,aaba,aabab,\ldots}$ is one such multiset. This example illustrates that it is possible to construct continuous measures on $\NS$.

\section{Axioms}
\label{sec:axioms}

Table \ref{tab:equations} contains several properties that allow for equational reasoning. An equation $e_1=e_2$ is \emph{sound} if $\semr{e_1}=\semr{e_2}$.
Probably the most counterintuitive is that sequential composition distributes over $\amp$ on the
right (Lemma \ref{lem:ampcmp}).
This is a consequence of the idea that probabilistic choices made by separate agents are independent.

\begin{table}
\caption{Equations}
\label{tab:equations}
\begin{tabular}{l@{\hspace{1cm}}l}
  \toprule
   $e_1 \amp e_2 = e_2 \amp e_1 $ & commutativity of $\amp$\\
   $e_1 \amp (e_2 \amp e_3) = (e_1 \amp e_2) \amp e_3$ & associativity of $\amp$\\
   $e \amp \fail = e$ & identity of $\amp$ \\
   $e_1 \oplus_r e_2 = e_2 \oplus_{1-r} e_1 $ & skew commutativity of $\oplus$\\
   $(e_1 \oplus_r e_2) \oplus_s e_3 = e_1 \oplus_{rs} (e_2 \oplus_{(s-rs)/(1-rs)} e_3)$ & skew associativity of $\oplus$\\
   $e_1 \oplus_1 e_2 = e_1$ & $\oplus$-elimination\\
   $e \oplus_r e = e$ & idempotence of $\oplus$\\
   $e_1 \cmp (e_2 \cmp e_3) = (e_1 \cmp e_2) \cmp e_3$ & associativity of $\cmp$\\
   $\skp \cmp e = e$ & left identity of $\cmp$ \\
   $e \cmp \skp = e$ & right identity of $\cmp$\\
   $\fail \cmp e = \fail$ & left absorption of $\cmp$\\
   $e \cmp \fail = \fail$ & right absorption of $\cmp$\\
   $(e_1 \amp e_2) \cmp e_3 = (e_1 \cmp e_3) \amp (e_2 \cmp e_3)$ & right distributivity of $\cmp$ over $\amp$\\
   $(e_1\oplus_r e_2)\cmp e_3 = (e_1\cmp e_3)\oplus_r (e_2\cmp e_3)$ & right distributivity of $\cmp$ over $\oplus_r$\\
   $(e_1\oplus_r e_2)\amp e_3 = (e_1\amp e_3)\oplus_r(e_2\amp e_3)$ & right distributivity of $\amp$ over $\oplus_r$\\
   $a\cmp(e_1\amp e_2) = (a\cmp e_1)\amp(a\cmp e_2)$ & atomic left distributivity of $\cmp$ over $\amp$\\
   $a\cmp(e_1\oplus_r e_2) = (a\cmp e_1)\oplus_r(a\cmp e_2)$ & atomic left distributivity of $\cmp$ over $\oplus_r$\\
   $d = e[d/x]\ \Iff\ d = \fix xe$ & fixpoint\\
  \bottomrule 
\end{tabular}
\end{table}

Soundness proofs for the equations listed in Table~\ref{tab:equations} are given in
\S\ref{sec:Apka}.

\begin{toappendix}
  \label{sec:Apka}
  
  This section contains soundness proofs for the properties listed in Table \ref{tab:equations} that are not obvious. Many of these proofs can be given in both denotational and operational style using the meta-calculus of \S\ref{sec:metacalculus}. We give proofs in both styles in most cases to showcase the versatility of the deductive system.
  
  In addition, we include an example illustrating the use of the axioms in calculating the probability of a compromised server in a network security scenario. 
  
  \begin{lemma}\ 
  \label{lem:collapse}
  \begin{itemize}
  \item Nondeterministic choice ($\amp$) is associative and commutative.
  \item Sequential composition ($\cmp$) is associative.
  \item $e_1\oplus_r e_2 = e_2\oplus_{1-r}e_1$.
  \item $(e_1\oplus_r e_2) \oplus_s e_3 = e_1\oplus_{rs} (e_2 \oplus_{(s-rs)/(1-rs)} e_3)$.
  \end{itemize}
  \end{lemma}
  
  \begin{lemma}
  \label{lem:ampcmp}
  $\semr{(e_1\amp e_2)\cmp e_3} = \semr{(e_1\cmp e_3)\amp(e_2\cmp e_3)}$.
  \end{lemma}
  \begin{proof}
  Using Lemma \ref{lem:elimcmp} and the definition of the substitution operator $[-]$ from \S\ref{sec:expressiontoautomaton},
  \begin{align*}
  \semr{(e_1\amp e_2)\cmp e_3} &= \semr{(e_1\amp e_2)[e_3]}
  = \semr{e_1[e_3]\amp e_2[e_3]}
  = \semr{(e_1\cmp e_3)\amp (e_2\cmp e_3)}.
  \end{align*}
  
  One can also argue operationally:
  \begin{align*}
  & \sample\semr{(e_1\amp e_2)\cmp e_3}\\
  &= \letin k{\sample\semr{e_1\amp e_2}}{\Sigma\mset{\sample x\cdot\semr{e_3}\mid x\in k}}\\
  &= \letin{(k_1,k_2)}{(\sample\semr{e_1},\sample\semr{e_2})}{}\\
  &\ieq \letin{k}{k_1+k_2}{\Sigma\mset{\sample x\cdot\semr{e_3}\mid x\in k}}\\
  &= \letin{(k_1,k_2)}{(\sample\semr{e_1},\sample\semr{e_2})}{}\\
  &\ieq \letin{(m_1,m_2)}{(\Sigma\mset{\sample x\cdot\semr{e_3}\mid x\in k_1},\Sigma\mset{\sample x\cdot\semr{e_3}\mid x\in k_2})}{m_1 + m_2}\\
  &= \letin{k_1}{\sample\semr{e_1}}{\letin{m_1}{\Sigma\mset{\sample x\cdot\semr{e_3}\mid x\in k_1}}{}}\\
  &\ieq \letin{k_2}{\sample\semr{e_2}}{\letin{m_2}{\Sigma\mset{\sample x\cdot\semr{e_3}\mid x\in k_2}}{m_1+m_2}}\\
  &= \letin{m_1}{\sample\semr{e_1\cmp e_3}}{\letin{m_2}{\sample\semr{e_2\cmp e_3}}{m_1+m_2}}\\
  &= \sample\semr{(e_1\cmp e_3)\amp(e_2\cmp e_3)}.
  \qedhere
  \end{align*}
  \end{proof}
  
  \begin{lemma}
  \label{lem:ocmp}
  $\semr{(e_1\oplus_r e_2)\cmp e_3} = \semr{(e_1\cmp e_3)\oplus_r (e_2\cmp e_3)}$.
  \end{lemma}
  \begin{proof}
  Using Lemma \ref{lem:elimcmp} and the definition of the substitution operator $[-]$ from \S\ref{sec:expressiontoautomaton},
  \begin{align*}
  \semr{(e_1\oplus_r e_2)\cmp e_3} &= \semr{(e_1\oplus_r e_2)[e_3]}
  = \semr{e_1[e_3]\oplus_r e_2[e_3]}
  = \semr{(e_1\cmp e_3)\oplus_r (e_2\cmp e_3)}.
  \end{align*}
  
  Operationally, one can get the same result given a random number generator $\rnd:()\to[0,1)$:
  \begin{align*}
  & \sample\semr{(e_1\oplus_r e_2)\cmp e_3}\\
  &= \letin k{\sample\semr{e_1\oplus_r e_2}}{\Sigma\mset{\sample x\cdot\semr{e_3}\mid x\in k}}\\
  &= \letin{(k_1,k_2)}{(\sample\semr{e_1},\sample\semr{e_2})}{}\\
  &\ieq \letin{k}{\mathsf{if\ } \rnd() \le r \mathsf{\ then\ } k_1 \mathsf{\ else\ } k_2}{\Sigma\mset{\sample x\cdot\semr{e_3}\mid x\in k}}\\
  &= \letin{(k_1,k_2)}{(\sample\semr{e_1},\sample\semr{e_2})}{}\\
  &\ieq \letin{(m_1,m_2)}{(\Sigma\mset{\sample x\cdot\semr{e_3}\mid x\in k_1},\Sigma\mset{\sample x\cdot\semr{e_3}\mid x\in k_2})}{}\\
  &\ieq \mathsf{if\ } \rnd() \le r \mathsf{\ then\ } m_1 \mathsf{\ else\ } m_2\\
  &= \letin{k_1}{\sample\semr{e_1}}{\letin{m_1}{\Sigma\mset{\sample x\cdot\semr{e_3}\mid x\in k_1}}{}}\\
  &\ieq \letin{k_2}{\sample\semr{e_2}}{\letin{m_2}{\Sigma\mset{\sample x\cdot\semr{e_3}\mid x\in k_2}}{}}\\
  &\ieq \mathsf{if\ } \rnd() \le r \mathsf{\ then\ } m_1 \mathsf{\ else\ } m_2\\
  &= \letin{m_1}{\sample\semr{e_1\cmp e_3}}{}\\
  &\ieq \letin{m_2}{\sample\semr{e_2\cmp e_3}}{\mathsf{if\ } \rnd() \le r \mathsf{\ then\ } m_1 \mathsf{\ else\ } m_2}\\
  &= \sample\semr{(e_1\cmp e_3)\oplus_r(e_2\cmp e_3)}.
  \qedhere
  \end{align*}
  \end{proof}
  
  \begin{lemma}
  \label{lem:cmpcmp}
  $\semr{(e_1\cmp e_2)\cmp e_3} = \semr{e_1\cmp(e_2\cmp e_3)}$.
  \end{lemma}
  \begin{proof}
  Using Lemma \ref{lem:elimcmp} and the definition of the substitution operator $[-]$ from \S\ref{sec:expressiontoautomaton},
  \begin{align*}
  \semr{(e_1\cmp e_2)\cmp e_3} &= \semr{(e_1\cmp e_2)[e_3]}
  = \semr{e_1\cmp(e_2[e_3])}
  = \semr{e_1\cmp(e_2\cmp e_3)}.
  \end{align*}
  
  Operationally, using the fact that $\sample w\cdot\sem{e} = w\cdot\sample\sem{e}$, we can get the same result:
  \begin{align*}
  & \sample\semr{(e_1\cmp e_2)\cmp e_3}\\
  &= \letin k{\sample\semr{e_1\cmp e_2}}{\Sigma\mset{\sample x\cdot\semr{e_3}\mid x\in k}}\\
  &= \letin{k_1}{\sample\semr{e_1}}{}\\
  &\ieq \letin{k}{\Sigma\mset{\sample y\cdot\semr{e_2}\mid y\in k_1}}{\Sigma\mset{\sample x\cdot\semr{e_3}\mid x\in k}}\\
  &= \letin{k_1}{\sample\semr{e_1}}{}\\
  &\ieq \Sigma\mset{\Sigma\mset{\sample x\cdot\semr{e_3}\mid x\in \sample y\cdot\semr{e_2}}\mid y\in k_1}\\
  &= \letin{k_1}{\sample\semr{e_1}}{}\\
  &\ieq \Sigma\mset{\Sigma\mset{\sample x\cdot\semr{e_3}\mid x\in y \cdot\sample\semr{e_2}}\mid y\in k_1}\\
  &= \letin{k_1}{\sample\semr{e_1}}{\letin{k_2}{\sample\semr{e_2}}{}}\\
  &\ieq \Sigma\mset{\Sigma\mset{\sample y\cdot x'\cdot\semr{e_3}\mid x'\in k_2}\mid y\in k_1}\\
  &= \letin{k_1}{\sample\semr{e_1}}{\letin{k_2}{\sample\semr{e_2}}{}}\\
  &\ieq \Sigma\mset{y\cdot\Sigma\mset{\sample x'\cdot\semr{e_3}\mid x'\in k_2}\mid y\in k_1}\\
  &= \letin{k_1}{\sample\semr{e_1}}{\Sigma\mset{y\cdot\sample\semr{e_2\cmp e_3}\mid y\in k_1}}\\
  &= \letin{k_1}{\sample\semr{e_1}}{\Sigma\mset{\sample y\cdot\semr{e_2\cmp e_3}\mid y\in k_1}}\\
  &= \sample\semr{e_1\cmp (e_2\cmp e_3)}.
  \qedhere
  \end{align*}
  \end{proof}
  
  \begin{lemma}
  \label{lem:oamp}
  $\semr{(e_1\oplus_r e_2)\amp e_3} = \semr{(e_1\amp e_3)\oplus_r(e_2\amp e_3)}$.
  \end{lemma}
  \begin{proof}
  The expression $(e_1\amp e_3)\oplus_r(e_2\amp e_3)$ can be obtained by applying the distributive law $\otimes: MD\to DM$ for finite multisets to $(e_1\oplus_r e_2)\amp e_3$. 
  \begin{align*}
  \begin{tikzpicture}[->, >=stealth', auto, node distance=32mm]
  \small
  \node (NW) {$MD\DNS$};
  \node (NM) [right of=NW] {$DM\DNS$};
  \node (NE) [right of=NM] {$DDM(\NS)$};
  \node (NEE) [right of=NE] {$D\DNS$};
  \node (SW) [below of=NW, node distance=12mm] {$M\DNS$};
  \node (SE) [below of=NE, node distance=12mm] {$DM(\NS)$};
  \node (SEE) [below of=NEE, node distance=12mm] {$\DNS$};
  \path (NW) edge node {$\otimes_{\DNS}$} (NM);
  \path (NW) edge node[swap] {$M\mu^D_{\NS}$} (SW);
  \path (NM) edge node[swap] {$D\otimes_{\NS}$} (NE);
  \path (NM) edge[bend angle=10, bend left] node {$D{\amp}_{\Sigma^*}$} (NEE);
  \path (NE) edge node[swap] {$DD\mu^M_{\Sigma^*}$} (NEE);
  \path (NE) edge node {$\mu^D_{M(\NS)}$} (SE);
  \path (NEE) edge node {$\mu^D_{\NS}$} (SEE);
  \path (SW) edge node {$\otimes_{\NS}$} (SE);
  \path (SW) edge[bend angle=10, bend right] node[swap] {${\amp}_{\Sigma^*}$} (SEE);
  \path (SE) edge node {$D\mu^M_{\Sigma^*}$} (SEE);
  \end{tikzpicture}
  \end{align*}
  The left-hand rectangle is an instance of the distributive law. The right-hand rectangle commutes by naturality; it is an instance of the horizontal product $\mu^D\mu^M$. The curved arrows are just the definition of ${\amp}_{\Sigma^*}$.
  
  The same result can be obtained operationally:
  \begin{align*}
  & \sample\semr{(e_1\oplus_r e_2)\amp e_3}\\
  &= \letin{(k, m_3)}{(\sample\semr{e_1\oplus_r e_2}, \sample\semr{e_3})}{k+m_3}\\
  &= \letin{(m_1, m_2)}{(\sample\semr{e_1}, \sample\semr{e_2})}{}\\
  &\ieq \letin{k}{\mathsf{if\ }\rnd()\le r\mathsf{\ then\ }m_1\mathsf{\ else\ }m_2}{}\\
  &\ieq \letin{m_3}{\sample\semr{e_3}}{k+m_3}\\
  &= \letin{m_1}{\sample\semr{e_1}}\\
  &\ieq \letin{m_2}{\sample\semr{e_2}}\\
  &\ieq \letin{m_3}{\sample\semr{e_3}}\\
  &\ieq \mathsf{if\ }\rnd()\le r\mathsf{\ then\ }m_1+m_3\mathsf{\ else\ }m_2+m_3\\
  &= \sample\semr{(e_1\amp e_3)\oplus_r(e_2\amp e_3)}.
  \tag*\qedhere
  \end{align*}
  \end{proof}
  
  \begin{lemma}
  \label{lem:pp}
  $\semr{a\cmp e} = a\cdot\semr{e} = \semr e\circ(a\cdot-)^{-1}$.
  \end{lemma}
  \begin{proof}
  By definition, we have $\semr{a\cmp e} = \semr{a}\bind-\cdot\semr{e} = a\cdot\semr{e}$, since $\semr{a} = \dirac{\mset a}$. Operationally, we can get the same result:
  \begin{align*}
  \sample\semr{a\cmp e} &= \letin{k}{\sample\semr{a}}{\Sigma\mset{\sample x\cdot\semr{e}\mid x\in k}}\\
  &= \letin{k}{a}{\Sigma\mset{\sample x\cdot\semr{e}\mid x\in k}}\\
  &= \sample a\cdot\semr{e}.
  \tag*\qedhere
  \end{align*}
  \end{proof}
  
  \begin{lemma}
  \label{lem:pe1e2}
  $\semr{a\cmp(e_1\amp e_2)} = \semr{(a\cmp e_1)\amp(a\cmp e_2)}$.
  \end{lemma}
  \begin{proof}
  Let $f=a\cdot-$, $\mu=\semr{e_1}$, and $\nu=\semr{e_2}$. Using Lemma \ref{lem:pp}, we wish to show
  \begin{align}
  (\mu\otimes\nu)\circ(+)^{-1}\circ f^{-1} = ((\mu\circ f^{-1})\otimes(\nu\circ f^{-1}))\circ(+)^{-1}.\label{eq:pe1e2}
  \end{align}
  First,
  \begin{align*}
  ((\mu\otimes\nu)\circ(f,g)^{-1})(A\times B)
  &= (\mu\otimes\nu)((f,g)^{-1})(A\times B))\\
  &= (\mu\otimes\nu)(f^{-1}(A)\times g^{-1}(B))\\
  &= (\mu\circ f^{-1})(A)\cdot(\nu\circ g^{-1})(B)\\
  &= ((\mu\circ f^{-1})\otimes(\nu\circ g^{-1}))(A\times B).
  \end{align*}
  As the two measures agree on measurable rectangles, we have
  \begin{align}
  (\mu\otimes\nu)\circ(f,g)^{-1} = (\mu\circ f^{-1})\otimes(\nu\circ g^{-1}).\label{eq:pe1e2b}
  \end{align}
  Since $f(m_1 + m_2) = f(m_1) + f(m_2)$, we have
  \begin{align}
  (+)^{-1}\circ f^{-1}
  = (f\circ(+))^{-1}
  = ((+)\circ(f,f))^{-1}
  = (f,f)^{-1}\circ(+)^{-1},\label{eq:pe1e2a}
  \end{align}
  so
  \begin{align*}
  (\mu\otimes\nu)\circ(+)^{-1}\circ f^{-1}
  &= (\mu\otimes\nu)\circ(f,f)^{-1}\circ(+)^{-1} && \text{by \eqref{eq:pe1e2a}}\\
  &= ((\mu\circ f^{-1})\otimes(\nu\circ f^{-1}))\circ(+)^{-1} && \text{by \eqref{eq:pe1e2b}.}
  \qedhere
  \end{align*}
  \end{proof}
  
  \begin{lemma}
  \label{lem:pe1e2x}
  $\semr{a\cmp(e_1\oplus_r e_2)} = \semr{(a\cmp e_1)\oplus_r(a\cmp e_2)}$.
  \end{lemma}
  \begin{proof}
  Let $f=a\cdot-$, $\mu=r\semr{e_1}$, and $\nu=(1-r)\semr{e_2}$. Using Lemma \ref{lem:pp}, we wish to show
  \begin{align}
  (\mu+\nu)\circ f^{-1} = \mu\circ f^{-1}+\nu\circ f^{-1}.\label{eq:pe1e2x}
  \end{align}
  For all measurable sets $A$,
  \begin{align*}
  ((\mu+\nu)\circ f^{-1})(A) 
  &= (\mu+\nu)(f^{-1}(A))\\
  &= \mu(f^{-1}(A))+\nu(f^{-1}(A))\\
  &= (\mu\circ f^{-1})(A)+(\nu\circ f^{-1})(A)\\
  &= ((\mu\circ f^{-1})+(\nu\circ f^{-1}))(A).
  \qedhere
  \end{align*}
  \end{proof}
  
  \begin{lemma}
  \label{lem:fixpoint}
  $\semr{d} = \semr{e\subst dx}\ \Iff\ \semr{d} = \semr{\fix xe}$.
  \end{lemma}
  \begin{proof}
  By Lemma \ref{lem:substrebind}, $\semr{e\subst dx} = \semr{e}[\semr d/x]$, and by Theorem \ref{thm:expsemwelldefined}, $\semr{\fix xe}$ is the unique fixpoint of the map $\mu\mapsto\semr{e}[\mu/x]$.
  \end{proof}
  
  \end{toappendix}

The axioms allow automata and expressions to be transformed to coalgebraic form
\begin{align}
\del : S\to D(\naturals\times(\naturals^S)^\Sigma),\label{eq:coalgebraicform}
\end{align}
as will be needed in \S\ref{sec:coalgebras}.

For automata in the form \eqref{eq:coalgebraicform}, $S$ is the set of states. Intuitively, for each $s\in S$, sampling $\del(s)$ results in an element of $\naturals\times(\Ns)^\Sigma$ whose projections
\begin{align*}
& \pi_\eps:\naturals\times(\naturals^S)^\Sigma\to\naturals && \pi_a:\naturals\times(\naturals^S)^\Sigma\to\naturals^S,\ a\in\Sigma
\end{align*}
give the multiplicity of accepting $\eps$ and strings beginning with $a$, respectively. Thus if sampling $\del(s)$ returns $v\in\naturals\times(\naturals^S)^\Sigma$ with some positive probability $p$, then starting from state $s$, with probability at least $p$, the automaton accepts the empty string $\eps$ starting from $s$ with multiplicity exactly $\pi_\eps(v)$, and $\pi_a(v)\in\naturals^S$ is the multiset of states occupied by an agent after scanning input symbol $a$. We say ``with probability at least $p$'' because other samples of $\del(s)$ may also contribute probability to these events.

\begin{examplerep}
\upshape
As an example of equational reasoning using the axioms of Table \ref{tab:equations}, we provide an equational proof of the correctness of the von Neumann trick for simulating a fair coin with a coin of arbitrary bias $r$:
  \begin{align}
    \fix{x}{((c \cmp x) \oplus_r a) \oplus_r (b \oplus_r (c \cmp x))} &= (\fix{x}{\skp \oplus_{2r(1-r)} (c \cmp x)}) \cmp (a \oplus_{1/2} b).
  \label{eq:vonNeumann}
  \end{align}
  The left-hand side says to flip the biased coin twice. If the result is heads-tails or tails-heads, output $a$ or $b$, respectively, otherwise repeat. The right-hand side says that $a$ and $b$ are produced with equal probability. The action $c$ counts the number of trials before success.

The proof can be found in the
appendix.
  \end{examplerep}
  \begin{proof}
  Abbreviating the left- and right-hand sides of \eqref{eq:vonNeumann} by $e_1$ and $e_2$, respectively, and writing $e_2 = f \cmp (a \oplus_{1/2} b)$ where $f = \fix{x}{\skp \oplus_{2r(1-r)}(c \cmp x)}$, the expression $e_2$ simplifies to
  \begin{align}
  e_2 &= f \cmp (a \oplus_{1/2} b) \nonumber \\
  &= (\skp \oplus_{2r(1-r)} (c \cmp f)) \cmp (a \oplus_{1/2} b)
  && \text{fixpoint} \nonumber\\
  &= (\skp \cmp (a \oplus_{1/2} b)) \oplus_{2r(1-r)} ((c \cmp f) \cmp (a \oplus_{1/2} b))
  && \text{right distributivity of $\cmp$ over $\oplus$} \nonumber\\
  &= (a \oplus_{1/2} b) \oplus_{2r(1-r)} ((c \cmp f) \cmp (a \oplus_{1/2} b)) && \text{identity of $\cmp$} \nonumber\\
  &= (a \oplus_{1/2} b) \oplus_{2r(1-r)} (c \cmp (f \cmp (a \oplus_{1/2} b)) && \text{associativity of $\cmp$} \nonumber\\
  &= (a \oplus_{1/2} b) \oplus_{2r(1-r)} (c \cmp e_2). \label{eq:vnfix}
  \end{align}
  We can similarly simplify $e_1$:
  \begin{align*}
  e_1 &= \fix{x}{((c \cmp x) \oplus_r a) \oplus_r (b \oplus_r (c \cmp x))} \\
  &= \fix{x}{(a \oplus_{1-r} (c \cmp x)) \oplus_r (b \oplus_r (c \cmp x))} && \text{skew-commutativity} \\
  &= \fix{x}{a \oplus_{r(1-r)} ((c \cmp x) \oplus_{r^2/(1-r+r^2)} (b \oplus_r (c \cmp x)))} && \text{skew-associativity} \\
  &= \fix{x}{a \oplus_{r(1-r)} ((b \oplus_r (c \cmp x)) \oplus_{(1-r)/(1-r+r^2)} (c \cmp x) )} && \text{skew-commutativity} \\
  &= \fix{x}{a \oplus_{r(1-r)} (b \oplus_{(r-r^2)/(1-r+r^2)} ((c \cmp x) \oplus_{(1-r)^2/(1-2r+2r^2)} (c \cmp x)))} && \text{skew-associativity} \\
  &= \fix{x}{a \oplus_{r(1-r)} (b \oplus_{(r-r^2)/(1- r + r^2)} (c \cmp x))} && \text{idempotence of $\oplus$} \\
  &= \fix{x}{(a \oplus_{1/2} b) \oplus_{2r(1-r)} (c \cmp x)}. && \text{skew-associativity}
  \end{align*}
  By the fixpoint axiom, $e_2 = e_1$ iff the following substitution holds:
  \begin{align*}
  e_2 &= (a \oplus_{1/2} b) \oplus_{2r(1-r)} (c \cmp x)[e_2/x]
  = (a \oplus_{1/2} b) \oplus_{2r(1-r)} (c \cmp e_2),
  \end{align*}
  which is just \eqref{eq:vnfix}.
\end{proof}

We do not know whether the axioms of Table \ref{tab:equations} are complete. However, they are sufficient to describe many behaviors we would like in a KAT-like system. Informally, one could augment expressions with tests by adding boolean expressions, which act as follows: $\semr{b} = \sem{\skp}$ if $b$ evaluates to \textsf{true} under $\rho$, and $\sem{\fail}$ otherwise. Basic control flow constructs for a simple imperative programming language can be implemented in the following way:
\begin{align*}
  & \ifelse{b}{e}{f} :=  (b \cmp e) \amp (\overline{b} \cmp f) && \while{b}{e} := \fix{f}{(b \cmp e \cmp f) \amp \overline{b}}
\end{align*}
The axioms of Table \ref{tab:equations} are sufficient to prove properties of this system, such as the following:
\begin{examplerep}
$\while{b}{e} = \ifelse{b}{(e \cmp \while{b}{e})}{\skp}$.
\end{examplerep}
\begin{proof}
Rewriting the left-hand side gives:
\begin{align*}
  \while{b}{e} &= \fix{f}{(b \cmp e \cmp f) \amp \overline{b}} && \text{definition of \textsf{while}}\\
  &= ((b \cmp e \cmp f) \amp \overline{b}) [(\fix{f}{(b \cmp e \cmp f) \amp \overline{b}}) / f] && \text{fixpoint} \\
  &= (b \cmp e \cmp (\fix{f}{(b \cmp e \cmp f) \amp \overline{b}})) \amp \overline{b} && \text{substitution} \\
  &= (b \cmp e \cmp (\while{b}{e})) \amp \overline{b} && \text{definition of \textsf{while}} \\
  &= (b \cmp e \cmp (\while{b}{e})) \amp (\overline{b} \cmp \skp) && \text{identity of $\cmp$} \\
  &= \ifelse{b}{(e \cmp \while{b}{e})}{\skp} && \text{definition of \textsf{if} }
\end{align*}
\end{proof}

A final example involving a network security scenario is given in
the appendix (Example \ref{exm:network}).
The example illustrates the use of the distributive law in reasoning about the combination of probability and nondeterminism.

\begin{toappendix}
\begin{example}
\label{exm:network}
\upshape
Here is an example involving a network security triage scenario. Say there is an attacker who wishes to initiate a day-0 exploit of a newly discovered network vulnerability. Suppose that there are $n$ servers on the network that require patching, but the network administrator can only patch $k$ of them in time and must leave $n-k$ vulnerable. On the other hand, the attacker has resources to mount an attack against only a small number of servers, say $m$. Given probabilistic strategies for the attacker and the administrator as to which servers to attack and defend, what is the likelihood that some server is compromised?

We can represent this as an expression of the form
\begin{align*}
\text{(defense strategy)}\amp\text{(attack strategy)}
\end{align*}
where (defense strategy) and (attack strategy) are expressions of the form
\begin{align*}
& {\oplus}\{\text{(choice $A$)} \mid A\subs\{1,2,\ldots,n\},\ \len A=k\}
&& {\oplus}\{\text{(choice $B$)} \mid B\subs\{1,2,\ldots,n\},\ \len B=m\},
\end{align*}
respectively, specifying probability distributions on the choices of which hosts to defend and which to attack, respectively, and (choice $A$) and (choice $B$) are expressions of the form
\begin{align*}
& {\amp}\mset{a_1,a_1,a_2,\seq a2k,a_k}
&& {\amp}\mset{\seq b1m},
\end{align*}
respectively, where $A = \{\seq a1k\}$ and $B = \{\seq b1m\}$. Note that elements of (choice $A$) occur with multiplicity two and those of (choice $B$) occur with multiplicity one. The probability that a server is compromised is the probability that the outcome multiset has an element of multiplicity one.

For example, for the case $n=3$, $k=2$, and $m=1$, we would have the expression of Fig.~\ref{fig:securityexample}, where the server names are $\{a,b,c\}$ and $p+q+r = u+v+w = 1$.
\begin{figure}
\begin{centering}
\begin{tikzpicture}[->, >=stealth', node distance=2cm, auto, bend angle=35, scale=0.8]
\small
 \node (top) at (0,-0.5) {$\amp$};
 \node (amp1) at (-3,-2) {$\oplus$};
 \path (top) edge (amp1);
 \node (amp2) at (3,-2) {$\oplus$};
 \path (top) edge (amp2);
 \node (ab) at (-5,-3.5) {$\amp$};
 \node (ac) at (-3,-3.5) {$\amp$};
 \node (bc) at (-1,-3.5) {$\amp$};
 \path (amp1) edge node[swap, yshift=-3pt] {$p$} (ab) edge node {$q$} (ac) edge node {$r$} (bc);
 \node (aba) at (-5.5,-5) {$2a$};
 \node (abb) at (-4.5,-5) {$2b$};
 \node (aca) at (-3.5,-5) {$2a$};
 \node (acc) at (-2.5,-5) {$2c$};
 \node (bcb) at (-1.5,-5) {$2b$};
 \node (bcc) at (-0.5,-5) {$2c$};
 \path (ab) edge (aba) edge (abb);
 \path (ac) edge (aca) edge (acc);
 \path (bc) edge (bcb) edge (bcc);
 \node (ea) at (1.5,-3.5) {$a$};
 \node (eb) at (3,-3.5) {$b$};
 \node (ec) at (4.5,-3.5) {$c$};
 \path (amp2) edge node[swap, yshift=-3pt] {$u$} (ea) edge node {$v$} (eb) edge node[yshift=-3pt] {$w$} (ec);
\end{tikzpicture}
\end{centering}
\caption{A network security problem}
\label{fig:securityexample}
\end{figure}
In this case the probability that a server is compromised is the dot product $(p,q,r)\cdot(w,v,u)$. If both parties choose their first option, which occurs with probability $up$, then the resulting multiset is $\mset{a,a,a,b,b}$, and no servers are compromised.

We show how the distributivity of nondeterminism over probability, in conjunction with the other axioms of Table~\ref{tab:equations}, can be used to calculate these probabilities. We first introduce some abbreviations for nested probabilistic and nondeterministic expressions. For $k\ge 1$, define
\begin{align}
{\amp}(\seq e1k) &= \begin{cases}
e_1 \amp ({\amp}(\seq e2k)), & \text{if $k > 1$}\\
e_1, & \text{if $k = 1$}
\end{cases}
\label{eq:ampdef}
\end{align}
and for $p_1 + \cdots + p_k = 1$,
\begin{align}
{\oplus}(e_1:p_1,\ldots,e_k:p_k) &= \begin{cases}
{\oplus}(e_2:p_2,\ldots,e_k:p_k), & \text{if $p_1 = 0$}\\
e_1 \oplus_{p_1} ({\oplus}(e_2:q_2,\ldots,e_k:q_k)), & \text{if $0 < p_1 < 1$}\\
e_1, & \text{if $p_1 = 1$},
\end{cases}
\label{eq:oplusdef}
\end{align}
where in the second case, $q_i = p_i/(p_2+\cdots+p_k) = p_i/(1-p_1)$ for $2\le i\le k$.
We further abbreviate ${\amp}(\seq e1k)$ and ${\oplus}(e_1:p_1,\ldots,e_k:p_k)$ by
\begin{align}
& {\amp}(e_i \mid 1\le i\le k) && {\oplus}(e_i:p_i \mid 1\le i\le k),
\label{eq:shortform}
\end{align}
respectively.

These abbreviations serve as syntactic sugar to allow us to reason with nondeterministic and probabilistic expressions with more than two subexpressions. The abbreviation for $\amp$ is justified by the associativity of $\amp$, and that for $\oplus$ is justified by the skew associativity of $\oplus$. One can transform any nested $\amp$ expression to the form \eqref{eq:ampdef} using the axioms of associativity and identity for $\amp$. Similarly, one can transform any nested $\oplus$ expression to the form \eqref{eq:oplusdef} using the axioms of skew associativity, skew commutativity, idempotence, and $\oplus$-elimination. The skew associativity and skew commutativity axioms allow the calculation of the correct probabilities.
One can also show that \eqref{eq:ampdef} and \eqref{eq:oplusdef} are independent of order; that is, for any permutation $\sigma$ of $\{1,2,\ldots,k\}$,
\begin{align}
{\amp}(e_i \mid 1\le i\le k) &= {\amp}(e_{\sigma(i)} \mid 1\le i\le k)\\
{\oplus}(e_i:p_i \mid 1\le i\le k) &= {\oplus}(e_{\sigma(i)}:p_{\sigma(i)} \mid 1\le i\le k).
\label{eq:permuteok}
\end{align}

In this notation, the distributive law for $\amp$ over $\oplus$ entails 
\begin{align}
& {\oplus}(e_1:p_1,e_2:p_2) \amp {\oplus}(d_1:q_1,d_2:q_2)\nonumber\\
&= {\oplus}(e_1\amp d_1:p_1q_1,e_1\amp d_2:p_1q_2,e_2\amp d_1:p_2q_1,e_2\amp d_2:p_2q_2),
\label{eq:newdistrib}
\end{align}
which allows probabilistic choices to precede nondeterministic choices. To see this,
\begin{align}
& {\oplus}(e_1:p_1,e_2:p_2) \amp {\oplus}(d_1:q_1,d_2:q_2)\nonumber\\
&= (e_1 \oplus_{p_1} e_2) \amp (d_1 \oplus_{q_1} d_2)\label{eq:diststepA}\\
&= (e_1 \amp (d_1 \oplus_{q_1} d_2)) \oplus_{p_1} (e_2 \amp (d_1 \oplus_{q_1} d_2))\label{eq:diststepB}\\
&= ((d_1 \oplus_{q_1} d_2)\amp e_1) \oplus_{p_1} ((d_1 \oplus_{q_1} d_2)\amp e_2)\label{eq:diststepC}\\
&= ((d_1\amp e_1) \oplus_{q_1} (d_2\amp e_1)) \oplus_{p_1} ((d_1\amp e_2) \oplus_{q_1} (d_2\amp e_2))\label{eq:diststepD}\\
&= ((e_1\amp d_1) \oplus_{q_1} (e_1\amp d_2)) \oplus_{p_1} ((e_2\amp d_1) \oplus_{q_1} (e_2\amp d_2))\label{eq:diststepE}\\
&= {\oplus}(e_1\amp d_1:p_1q_1,e_1\amp d_2:p_1q_2,e_2\amp d_1:p_2q_1,e_2\amp d_2:p_2q_2),\label{eq:diststepF}
\end{align}
where
step \eqref{eq:diststepA} is justified by \eqref{eq:oplusdef};
\eqref{eq:diststepB} and \eqref{eq:diststepD} by distributivity;
\eqref{eq:diststepC} and \eqref{eq:diststepE} by commutativity; and
\eqref{eq:diststepF} by skew associativity, \eqref{eq:oplusdef}, and calculation.
The law \eqref{eq:newdistrib} further generalizes in a straightforward way to
\begin{align}
({\oplus}(e_i:p_i \mid 1\le i\le k) \amp ({\oplus}(d_j:q_j \mid 1\le j\le m)
&= {\oplus}(e_i\amp d_j:p_iq_j \mid 1\le i\le k,\ 1\le j\le m).
\label{eq:genoplus}
\end{align}

In our application, we wish to calculate the probability that a server is compromised, which is the probability that the expression of Fig.~\ref{fig:securityexample} generates a multiset containing an element with multiplicity 1. In the notation of \eqref{eq:ampdef} and \eqref{eq:oplusdef}, this expression is
\begin{align*}
{\oplus}({\amp}(a,a,b,b):p,{\amp}(a,a,c,c):q,{\amp}(b,b,c,c):r) \amp {\oplus}(a:u,b:v,c:w)
\end{align*}
or allowing integer coefficients to abbreviate multiplicities,
\begin{align*}
{\oplus}({\amp}(2a,2b):p,{\amp}(2a,2c):q,{\amp}(2b,2c):r) \amp {\oplus}(a:u,b:v,c:w).
\end{align*}
To avoid dealing with corner cases, let us assume that all $p,q,r,u,v,w$ are nonzero. This is without loss of generality, as zero probabilities can be systematically eliminated and do not appear in the unsugared form of expressions \eqref{eq:oplusdef}. Now
\begin{align*}
& {\oplus}({\amp}(2a,2b):p,{\amp}(2a,2c):q,{\amp}(2b,2c):r) \amp {\oplus}(a:u,b:v,c:w)\\
&= {\oplus}({\amp}(2a,2b)\amp a:pu,{\amp}(2a,2c)\amp a:pv,{\amp}(2b,2c)\amp a:pw,\\
& \hspace{8.3mm} {\amp}(2a,2b)\amp b:qu,{\amp}(2a,2c)\amp b:qv,{\amp}(2b,2c)\amp b:qw,\\
& \hspace{8.3mm} {\amp}(2a,2b)\amp c:ru,{\amp}(2a,2c)\amp c:rv,{\amp}(2b,2c)\amp c:rw) && \text{by \eqref{eq:genoplus}}\\
&= {\oplus}({\amp}(3a,2b):pu,{\amp}(3a,2c):pv,{\amp}(a,2b,3c):pw,\\
& \hspace{8.3mm} {\amp}(2a,3b):qu,{\amp}(2a,b,2c):qv,{\amp}(3b,2c):qw,\\
& \hspace{8.3mm} {\amp}(2a,2b,c):ru,{\amp}(2a,3c):rv,{\amp}(2b,3c):rw) && \text{by commutativity and \eqref{eq:ampdef}}\\
&= {\oplus}(
{\oplus}({\amp}(3a,2b):pu/s,{\amp}(3a,2c):pv/s,{\amp}(2a,3b):qu/s,\\
& \hspace{12mm}{\amp}(3b,2c):qw/s,{\amp}(2a,3c):rv/s,{\amp}(2b,3c):rw/s : s,\\
& \hspace{8.3mm}{\oplus}({\amp}(a,2b,3c):pw/(1-s),{\amp}(2a,b,2c):qv/(1-s),\\
& \hspace{12.2mm}{\amp}(2a,2b,c):ru/(1-s)): (1-s)) && \text{skew associativity, calculation}\\
&= ({\oplus}({\amp}(3a,2b):pu/s,{\amp}(3a,2c):pv/s,{\amp}(2a,3b):qu/s,\\
& \hspace{9.6mm}{\amp}(3b,2c):qw/s,{\amp}(2a,3c):rv/s,{\amp}(2b,3c):rw/s : s))\\
& \hspace{5mm}\oplus_s ({\oplus}({\amp}(a,2b,3c):pw/(1-s),{\amp}(2a,b,2c):qv/(1-s),\\
& \hspace{16mm}{\amp}(2a,2b,c):ru/(1-s)): 1-s) && \text{by \eqref{eq:oplusdef},}
\end{align*}
where
\begin{align*}
s &= pu + pv + qu + qw + rv + rw
= (p,q,r)\cdot((u,w,v) + (v,u,w))\\
1-s &= pw + qv + ru = (p,q,r)\cdot(w,v,u).
\end{align*}
The subexpression
\begin{align*}
{\oplus}({\amp}(a,2b,3c):pw/(1-s),{\amp}(2a,b,2c):qv/(1-s),{\amp}(2a,2b,c):ru/(1-s)),
\end{align*}
which occurs with probability $1-s$, describes the undesirable set of outcomes containing an element with multiplicity 1.
\end{example}
\end{toappendix}

\section{Metric Properties}
\label{sec:metric}

To show that the semantic maps for automata and expressions are well defined, we will introduce a complete ultrametric on $\DNS$ and show that the semantic equations are contractive, thus have a unique solution. The metric is also useful for many other purposes, such as definitions and proofs by coinduction.

Recall from \S\ref{sec:basics} the definition
\begin{align*}
\alpha \equiv_n \beta\ &\Iff\ \alpha(x)=\beta(x)\ \text{for all $\len x\le n$} && [\alpha]_n = \set{\beta}{\beta\equiv_n\alpha}
\end{align*}
for $\alpha,\beta\in\NS$. By convention, we take $\equiv_{-1}$ to be the trivial relation with one equivalence class. The relations $\equiv_{n}$ are ordered by refinement, with $\bigcap_n\equiv_n$ the identity relation. Every $\equiv_n$-equivalence class $[\alpha]_n$ has a unique canonical element $\alpha\rest n$ whose support is contained in $\Sigma^{\le n}$,
and $\alpha\equiv_n\beta$ iff $\alpha\rest n=\beta\rest n$.

\begin{lemmarep}
\label{lem:integraldomain}
For $n\ge 0$, $\alpha\in\NS$, and $\mu,\nu\in\DNS$,
\begin{align*}
(\mu\amp\nu)([\alpha]_n) = \sum_{\beta+\gamma=\alpha\rest n}\mu([\beta]_n)\cdot\nu([\gamma]_n).
\end{align*}
More generally, let $m$ be a finite multiset of distributions in $\DNS$. For $n\ge 0$ and $\alpha\in\NS$,
\begin{align*}
({\amp}m)([\alpha]_n) = \sum_{\alpha\rest n=\sum_{\mu\in m}\beta_\mu} \ \prod_{\mu\in m}\mu([\beta_\mu]_n)
\end{align*}
where the sum is over all possible ways of associating a multiset $\beta_\mu$ with each (occurrence of) $\mu\in m$ such that
$\alpha\rest n=\sum_{\mu\in m}\beta_\mu$.
\end{lemmarep}
\begin{proof}
We prove the first statement, which is the binary case. The second statement is a direct generalization. Note that in order for $\beta+\gamma=\alpha\rest n$, we must have $\beta=\beta\rest n$ and $\gamma=\gamma\rest n$.
\begin{align*}
(\mu\amp\nu)([\alpha]_n) &= ((\mu\otimes\nu)\circ(+)^{-1})([\alpha]_n)\\
&= (\mu\otimes\nu)(\set{(\beta,\gamma)}{\beta+\gamma\equiv_n\alpha})\\
&= (\mu\otimes\nu)(\set{(\beta,\gamma)}{\beta\rest n+\gamma\rest n=\alpha\rest n})\\
&= \sum_{\beta+\gamma=\alpha\rest n}(\mu\otimes\nu)([\beta]_n\times[\gamma]_n)
= \sum_{\beta+\gamma=\alpha\rest n}\mu([\beta]_n)\cdot\nu([\gamma]_n).
\tag*\qedhere
\end{align*}
\end{proof}

\subsection{A Complete Ultrametric}
\label{sec:ultrametric}

We define an equivalence relation on $\DNS$, also denoted $\equiv_n$.
\begin{align*}
\mu\equiv_n\nu\ &\Iff\ \forall\alpha\in\NS\ \mu([\alpha]_n)=\nu([\alpha]_n).
\end{align*}
This gives rise to a complete ultrametric on $\DNS$.
\begin{align*}
d(\mu,\nu) &= \begin{cases}
2^{-n}, & \text{if $n$ is minimum such that $\mu\not\equiv_n\nu$,}\\
0, & \text{if no such $n$ exists.}
\end{cases}
\end{align*}
It follows from the definition that
\begin{align}
d(\mu,\nu) \le 2^{-(n+1)}\ \Iff\ \mu\equiv_n\nu\ \Iff\ 
\forall \alpha\in\NS\ \mu([\alpha]_n)=\nu([\alpha]_n).\label{eq:ultrametricY}
\end{align}

\begin{lemmarep}
\label{lem:ultrametric}
The map $d:\DNS^2\to\reals$ is a complete ultrametric.
\end{lemmarep}
\begin{proof}
It is routine to show that it is a pseudometric. If $d(\mu,\nu)=0$, then $\mu([\alpha]_n)=\nu([\alpha]_n)$ for all $\alpha\in\NS$ and $n\ge 0$. But these sets generate the Borel space, so $\mu=\nu$. Thus $d$ is a metric.

Completeness is also not difficult to show. Let $\mu_i$, $i\ge 0$ be a Cauchy sequence. Given $n$, for all sufficiently large $i,j$, $d(\mu_i,\mu_j) < 2^{-n}$. By \eqref{eq:ultrametricY}, $\mu_i$ and $\mu_j$ agree on $[\alpha]_n$ for all $\alpha\in\NS$, hence $\mu_i([\alpha]_n)$ stabilizes at a constant value for sufficiently large $i$. We define $\mu([\alpha]_n)$ to be this value. Then $\mu$ is a measure by the Kolmogorov extension theorem, and $d(\mu_i,\mu) < 2^{-n}$ for all sufficiently large $i$, so the $\mu_i$ converge to $\mu$.
\end{proof}

\begin{lemmarep}\ 
\label{lem:ultrametricB}
\begin{enumerate}[\upshape(i)]
\item
$d(x\cdot \mu,x\cdot \nu) \le 2^{-\len x}d(\mu,\nu)$
\item
$d(\sum_i r_i \mu_i,\sum_i r_i \nu_i) \le \max_i d(\mu_i,\nu_i)$
\item
$d({\amp}(\mset{\mu_i \mid i\in m}),{\amp}(\mset{\nu_i \mid i\in m})) \le \max_{i\in m} d(\mu_i,\nu_i)$.
\item
$d(\mu_1 \bind -\cdot\nu_1,\mu_2 \bind -\cdot\nu_2) \le \max d(\mu_1,\mu_2), d(\nu_2,\nu_2)$.
\end{enumerate}
\end{lemmarep}
\begin{proof}
(i) By definition, $(x\cdot \beta)(xy) = \beta(y)$ and $(x\cdot\beta)(z) = 0$ if $x$ is not a prefix of $z$. Thus
\begin{align*}
(x\cdot-)^{-1}([x\cdot\alpha]_{n+\len x}) &= \set{\beta}{x\cdot\beta\in[x\cdot \alpha]_{n+\len x}}
= \set{\beta}{x\cdot\beta\equiv_{n+\len x}x\cdot\alpha}
= \set{\beta}{\beta\equiv_{n}\alpha}
= [\alpha]_n,
\end{align*}
and for $\gamma$ not of the form $x\cdot\alpha$,
\begin{align*}
(x\cdot-)^{-1}([\gamma]_{n+\len x}) &= \set{\beta}{x\cdot\beta\in[\gamma]_{n+\len x}}
= \set{\beta}{x\cdot\beta\equiv_{n+\len x}\gamma}
= \emptyset.
\end{align*}
Combining these,
\begin{align*}
(x\cdot\mu)([\gamma]_{n+\len x})
&= \mu((x\cdot-)^{-1}([\gamma]_{n+\len x}))
= \begin{cases}
\mu([\alpha]_n), & \text{if $\gamma = x\cdot\alpha$,}\\
0, & \text{if $\gamma$ is not of the form $x\cdot\alpha$.}
\end{cases}
\end{align*}
Thus for any $n$, $x\cdot\mu$ and $x\cdot\nu$ agree on all $[\alpha]_{n+\len x}$ iff $\mu$ and $\nu$ agree on all $[\alpha]_n$. By \eqref{eq:ultrametricY}, $d(x\cdot\mu,x\cdot\nu)<2^{n+\len x}$ iff $d(\mu,\nu)<2^n$, so $d(x\cdot\mu,x\cdot\nu)\le 2^{-\len x}d(\mu,\nu)$.

(ii) Let $2^{-n}=\max_i d(\mu_i,\nu_i)$. Then $d(\mu_i,\nu_i)\le 2^{-n}$ for all $i$, thus by \eqref{eq:ultrametricY}, $\mu_i$ and $\nu_i$ agree on $[\alpha]_{n-1}$ for all $\alpha\in\NS$. Then $\sum_i r_i \mu_i$ and $\sum_i r_i \nu_i$ also agree on all $[\alpha]_{n-1}$, so $d(\sum_i r_i \mu_i,\sum_i r_i \nu_i) \le 2^{-n}$.

(iii) Let $2^{-n}=\max_i d(\mu_i,\nu_i)$. Then $d(\mu_i,\nu_i)\le 2^{-n}$ for all $i\in I$, thus by \eqref{eq:ultrametricY}, $\mu_i$ and $\nu_i$ agree on $[\alpha]_{n-1}$ for all $\alpha\in\NS$. By Lemma \ref{lem:integraldomain}, ${\amp}(\mset{\mu_i \mid i\in I}$ and ${\amp}(\mset{\nu_i \mid i\in I})$ also agree on all $[\alpha]_{n-1}$, so $d({\amp}(\mset{\mu_i \mid i\in I},{\amp}(\mset{\nu_i \mid i\in I})) \le 2^{-n}$.

(iv) Recall that the bind operation $\mu\bind-\cdot\nu$ is defined by first extending $-\cdot\nu:\Sigma^*\to\DNS$ to $-\cdot\nu:\NS\to\DNS$, then integrating with respect to $\mu$. Given a measurable set $A\subs\NS$,
\begin{align*}
(\mu\bind -\cdot\nu)(A) = \int_{\beta\in\NS}\ (\beta\cdot\nu)(A)\,d\mu.
\end{align*}
Since $\beta\cdot\nu\equiv_n\beta\rest n\cdot\nu$, when applied to $[\alpha]_n$ we can express the integral as a countable sum:
\begin{align*}
(\mu\bind -\cdot\nu)([\alpha]_n) &= \int_{\beta\in\NS}\ (\beta\cdot\nu)([\alpha]_n)\,d\mu\\
&= \int_{\beta\in\NS}\ (\beta\rest n\cdot\nu)([\alpha]_n)\,d\mu
= \sum_{\beta=\beta\rest n}\ (\beta\cdot\nu)([\alpha]_n)\,\mu([\beta]_n).
\end{align*}
Thus if $\mu_1\equiv_n\mu_2$ and $\nu_1\equiv_n\nu_2$, then
\begin{align*}
(\mu_1\bind -\cdot\nu_1)([\alpha]_n)
&= \sum_{\beta=\beta\rest n}\ (\beta\cdot\nu_1)([\alpha]_n)\,\mu_1([\beta]_n)\\
&= \sum_{\beta=\beta\rest n}\ (\beta\cdot\nu_2)([\alpha]_n)\,\mu_2([\beta]_n)
= (\mu_2\bind -\cdot\nu_2)([\alpha]_n),
\end{align*}
therefore $\mu_1\bind -\cdot\nu_1 \equiv_n \mu_2\bind -\cdot\nu_2$. The stated result (iv) follows.
\end{proof}

\subsection{A Metric on Labelings}
\label{sec:metriconlabelings}

Let $X$ be a set. Consider labelings $L:X\to\DNS$ of elements of $X$ with distributions over $\NS$.
Let us lift the metric $d$ on $\DNS$ to labelings as follows:
\begin{align*}
d(L_1,L_2) &= \sup_{x\in X} d(L_1(x),L_2(x)).
\end{align*}
We might just as well write $\max$ instead of $\sup$ because the supremum is always achieved, as values are of the form $2^{-n}$ or $0$.
Let us also extend the equivalence relations $\equiv_n$ on $\DNS$ to labelings:
\begin{align*}
L_1\equiv_{n}L_2\ &\Iff\ \forall x\in X\ L_1(x)\equiv_{n}L_2(x).
\end{align*}

\begin{lemmarep}\ 
\label{lem:metricproperties}
\begin{enumerate}[{\upshape(i)}]
\item
$d:(X\to\DNS)^2\to\reals$ is a complete ultrametric.
\item
$d(L_1,L_2) \le 2^{-(n+1)}\ \Iff\ L_1\equiv_n L_2$.
\item
$d(L_1,L_2) = \begin{cases}
2^{-n} & \text{if $n$ is minimum such that $L_1\not\equiv_n L_2$}\\
0 & \text{if no such $n$ exists.}
\end{cases}$
\end{enumerate}
\end{lemmarep}
\begin{proof}
The statement (i) is a standard construction based on Lemma \ref{lem:ultrametric}. We argue the (ultrametric) triangle inequality and completeness explicitly. For any $x\in X$,
\begin{align*}
d(L_1(x),L_2(x)) &\le \max d(L_1(x),L_3(x)),d(L_3(x),L_2(x))\\
&\le \max\sup_{y\in X} d(L_1(y),L_3(y)),\sup_{z\in X}d(L_3(z),L_2(z))\\
&= \max d(L_1,L_3),d(L_3,L_2).
\end{align*}
As $x\in X$ was arbitrary,
\begin{align*}
d(L_1,L_2) &= \sup_{x\in X} d(L_1(x),L_2(x)) \le \max d(L_1,L_3),d(L_3,L_2).
\end{align*}
For completeness, let $L_i$, $i\ge 0$ be a Cauchy sequence. Then $L_i(x)$ is Cauchy for each $x\in X$. By Lemma \ref{lem:ultrametric}, $L_i(x)$ converges to a value, which we will call $L(x)$. Moreover, because the extension is defined as a supremum over $x$, the convergence is uniform in $x$, thus $L_i$ converges to $L$.

For (ii),
\begin{align*}
d(L_1,L_2) \le 2^{-(n+1)}\ &\Iff\ \max_{x\in X} d(L_1(x),L_2(x)) \le 2^{-(n+1)}\\
&\Iff\ \forall x\in X\ d(L_1(x),L_2(x)) \le 2^{-(n+1)}\\
&\Iff\ \forall x\in X\ L_1(x)\equiv_n L_2(x))\\
&\Iff\ L_1\equiv_n L_2.
\end{align*}

The statement (iii) follows from (ii) and the fact that all values are of the form $2^{-n}$ or $0$.
\end{proof}

\subsection{Verification of the Semantics}

The results of \S\ref{sec:ultrametric} and \S\ref{sec:metriconlabelings} allow us to create semantic maps by a fixpoint construction. Consider an automaton with finite state set $S$. Recalling the coinductive definition of $\sem s$ from \S\ref{sec:semanticsofautomata}, we define a map $\tau:(S\to\DNS)\to(S\to\DNS)$ and show that it is eventually contractive with constant of contraction $1/2$.

Given $L:S\to\DNS$,
let
\begin{align}
\tau(L)(s) &= \begin{cases} 
\dirac{\mset\eps} & \text{if $\ell(s)=\skp$}\\
\dirac{\mset{}} & \text{if $\ell(s)=\fail$}\\
a \cdot L(t) & \text{if $\ell(s)= a\in\Sigma$ and $\del(s)=t$}\\
\sum_i r_i L(t_i) & \text{if $\ell(s)=\oplus$ and $\del(s)=\sum_i r_it_i$}\\
{\amp}(ML(m)) & \text{if $\ell(s)={\amp}$ and $\del(s)=m$.}
\end{cases}
\label{eq:autotaudef}
\end{align}

\begin{lemmarep}
\label{lem:autcontractive}
The map $\tau:(S\to\DNS)\to(S\to\DNS)$ of \eqref{eq:autotaudef} is eventually contractive with (eventual) constant of contraction $1/2$: for all $L_1,L_2$, for any $n\ge\len S$,
\begin{align*}
d(\tau^n(L_1),\tau^n(L_2)) \le \textstyle\frac 12d(L_1,L_2).
\end{align*}
\end{lemmarep}
\begin{proof}
By the productivity assumption, every path from any state $s$ visits a terminal or action state within $\len S$ steps. 
We show by induction on $k$ that if $k$ is the least number such that all paths from $s$ have visited a terminal or action state within $k$ steps, then $d(\tau^{k+1}(L_1)(s),\tau^{k+1}(L_2)(s)) \le \textstyle\frac 12d(L_1,L_2)$. There are five cases:
\begin{itemize}
\item 
If $\ell(s)=\skp$, then $k=0$ and $\tau(L_1)(s) = \tau(L_2)(s) = \dirac{\mset\eps}$, so $d(\tau(L_1)(s),\tau(L_2)(s)) = 0$.
\item 
If $\ell(s)=\fail$, then $k=0$ and $\tau(L)(s) = \tau(L_2)(s) = \dirac{\mset{}}$, so $d(\tau(L_1)(s),\tau(L_2)(s)) = 0$.
\item 
If $\ell(s)=a\in\Sigma$ and $\del(s)=t$, then $k=0$, $\tau(L_1)(s) = a\cdot L_1(t)$, and $\tau(L_2)(s) = a\cdot L_2(t)$, so
\begin{align*}
d(\tau(L_1)(s),\tau(L_2)(s)) &= d(a\cdot L_1(t),a\cdot L_2(t))\\
&\le \textstyle\frac 12d(L_1(t),L_2(t)) & \text{by Lemma \ref{lem:ultrametricB}(i)}\\
&\le \textstyle\frac 12\max_{u\in S} d(L_1(u),L_2(u))\\
&= \textstyle\frac 12d(L_1,L_2).
\end{align*}
\item 
If $\ell(s)=\oplus$ and $\del(s)=\sum_i r_it_i$, then $k>0$, $\tau(L_1)(s) = \sum_i r_i L_1(t_i)$, and $\tau(L_2)(s) = \sum_i r_i L_2(t_i)$, so
\begin{align*}
d(\tau^{k+1}(L_1)(s),\tau^{k+1}(L_2)(s))
&= d(\tau(\tau^k(L_1)(s)),\tau(\tau^k(L_2)(s)))\\
&= d(\sum_i r_i \tau^k(L_1)(t_i),\sum_i r_i \tau^k(L_2)(t_i))\\
&\le \max_i d(\tau^k(L_1)(t_i),\tau^k(L_2)(t_i)) && \text{by Lemma \ref{lem:ultrametricB}(ii)}\\
&\le \max_i \textstyle\frac 12d(L_1,L_2) && \text{induction hypothesis}\\
&= \textstyle\frac 12d(L_1,L_2).
\end{align*}
\item 
If $\ell(s)={\amp}$ and $\del(s)=m$, then $k>0$, $\tau(L_1)(s) = {\amp}(ML_1(m))$, and $\tau(L_2)(s) = {\amp}(ML_2(m))$, so
\begin{align*}
d(\tau^{k+1}(L_1)(s),\tau^{k+1}(L_2)(s))
&= d(\tau(\tau^k(L_1)(s)),\tau(\tau^k(L_2)(s)))\\
&= d({\amp}(M(\tau^k(L_1))(m)),{\amp}(M(\tau^k(L_2))(m)))\\
&\le \max_{t\in m} d(\tau^k(L_1)(t),\tau^k(L_2)(t)) && \text{by Lemma \ref{lem:ultrametricB}(iii)}\\
&\le \max_{t\in m} \textstyle\frac 12d(L_1,L_2) && \text{induction hypothesis}\\
&= \textstyle\frac 12d(L_1,L_2).
\end{align*}
\end{itemize}
\end{proof}

\begin{theoremrep}
\label{thm:autsemwelldefined}
The semantic map $\sem-$ on automata is well defined.
\end{theoremrep}
\begin{proof}
The definition of $\sem-$ says exactly that $\sem-$ is a fixpoint of $\tau$; that is, $\tau(\sem-)=\sem-$. As $\tau$ is an eventually contractive map on a complete metric space, it has a unique fixpoint by the Banach fixpoint theorem, which must be $\sem-$.
\end{proof}

The semantics of expressions from \S\ref{sec:semanticsofexpressions} can be handled similarly. The definition of $\sem-$ from \S\ref{sec:semanticsofexpressions} asks for a fixpoint of
\begin{align*}
\tau:(\Exp\times(\Var\to\DNS)\to\DNS) \to (\Exp\times(\Var\to\DNS)\to\DNS),
\end{align*}
where for $L:\Exp\times(\Var\to\DNS)\to\DNS$,
\begin{align}
\tau(L)(e,\rho) &= \begin{cases}
\dirac{\mset\eps} & \text{if $e=\skp$}\\
\dirac{\mset{}} & \text{if $e=\fail$}\\
\dirac{\mset a} & \text{if $e=a\in\Sigma$}\\
\rho(x) & \text{if $e=x\in\Var$}\\
rL(e_1,\rho) + (1-r)L(e_2,\rho) & \text{if $e=e_1 \oplus_r e_2$}\\
L(e_1,\rho)\amp L(e_2,\rho) = (L(e_1,\rho)\otimes L(e_2,\rho))\circ(+)^{-1} & \text{if $e=e_1 \amp e_2$}\\
L(e_1,\rho) \bind -\cdot L(e_2,\rho) & \text{if $e=e_1 \cmp e_2$}\\
L(e_1,\rho[L(\fix x{e_1},\rho)/x]) & \text{if $e=\fix x{e_1}$.}
\end{cases}
\label{eq:exptaudef}
\end{align}

Let $t\in\Var\cup\{\skp\}$. We say that an occurrence of $t$ in $e$ is \emph{unguarded} in $e$ if either
\begin{itemize}
\item
$e=t$;
\item
$e\in\{e_1\oplus_r e_2,e_1\amp e_2\}$ and $t$ is unguarded in either $e_1$ or $e_2$;
\item
$e=\fix xd$, $t\ne x$, and $t$ is unguarded in $d$; or
\item
$e=e_1\cmp e_2$, $t$ is unguarded in $e_2$, and $\skp$ is unguarded in $e_1$,
\end{itemize}
otherwise $t$ is \emph{guarded} in $e$.
Define
\begin{align*}
U(e) &= \{\text{unguarded free variables of $e$}\} &
G(e) &= \{\text{guarded free variables of $e$}\}.
\end{align*}
For $\rho_1,\rho_2\in\Var\to\DNS$ and $V\subs\Var$, let
\begin{align*}
d_{V}(\rho_1,\rho_2) &= \max_{x\in V} d(\rho_1(x),\rho_2(x)) &
d_e(\rho_1,\rho_2) &= \max \{ d_{U(e)}(\rho_1,\rho_2),\,\half d_{G(e)}(\rho_1,\rho_2) \}.
\end{align*}

\begin{lemmarep}
\label{lem:expcontractiveA}
For a fixed finite set $A\subs\Exp$ and $L_1,L_2\in\Exp\to(\Var\to\DNS)\to\DNS$, let
\begin{align*}
d(L_1,L_2) & = \sup\set{d(L_1(e,\rho),L_2(e,\rho))}{e\in A,\ \rho:\Var\to\DNS}.
\end{align*}
For sufficiently large $n$ depending on the height of $e$, for any $\rho_1,\rho_2$,
\begin{align*}
d(\tau^n(L_1)(e,\rho_1),\tau^n(L_2)(e,\rho_2)) &\le \max d_e(\rho_1,\rho_2),\,\half d(L_1,L_2).
\end{align*}
\end{lemmarep}
\begin{proof}
We proceed by induction. For the basis, for all $n\ge 1$,
\begin{align*}
& d(\tau^n(L_1)(\skp,\rho_1),\tau^n(L_2)(\skp,\rho_2)) = d(\dirac{\mset\eps},\dirac{\mset\eps}) = 0\\
& d(\tau^n(L_1)(\fail,\rho_1),\tau^n(L_2)(\fail,\rho_2)) = d(\dirac{\mset{}},\dirac{\mset{}}) = 0\\
& d(\tau^n(L_1)(a,\rho_1),\tau^n(L_2)(a,\rho_2)) = d(\dirac{\mset a},\dirac{\mset a}) = 0\\
& d(\tau^n(L_1)(x,\rho_1),\tau^n(L_2)(x,\rho_2)) = d(\rho_1(x),\rho_2(x)) \le d_x(\rho_1,\rho_2).
\end{align*}

For $e_1 \oplus_r e_2$, by choosing $n$ sufficiently large, the induction hypothesis gives
\begin{align*}
d(\tau^n(L_1)(e_i,\rho_1),\tau^n(L_2)(e_i,\rho_2)) &\le \max d_{e_i}(\rho_1,\rho_2),\,\half d(L_1,L_2),\ \ i\in\{1,2\}.
\end{align*}
Using Lemma \ref{lem:ultrametricB}(ii),
\begin{align*}
& d(\tau^{n+1}(L_1)(e_1 \oplus_r e_2,\rho_1),\tau^{n+1}(L_2)(e_1 \oplus_r e_2,\rho_2))\\
&= d(\tau^n(L_1)(e_1,\rho_1)\oplus_r \tau^n(L_1)(e_2,\rho_1),\tau^n(L_2)(e_1,\rho_2)\oplus_r \tau^n(L_2)(e_2,\rho_2))\\
&\le \max d(\tau^n(L_1)(e_1,\rho_1),\tau^n(L_2)(e_1,\rho_2)),\ d(\tau^n(L_1)(e_2,\rho_1),\tau^n(L_2)(e_2,\rho_2))\\
&\le \max d_{e_1}(\rho_1,\rho_2),\ d_{e_2}(\rho_1,\rho_2),\,\half d(L_1,L_2)\\
&\le \max d_{e_1 \oplus_r e_2}(\rho_1,\rho_2),\,\half d(L_1,L_2),
\end{align*}
since $G(e_1), G(e_2) \subs G(e_1 \oplus_r e_2)$ and
$U(e_1), U(e_2) \subs U(e_1 \oplus_r e_2)$.
The same argument applies to $e_1\amp e_2$ using Lemma \ref{lem:ultrametricB}(iii).

For $e_1 \cmp e_2$, let $n$ be sufficiently large that
\begin{align*}
d(\mu_1,\mu_2) &\le \max d_{e_1}(\rho_1,\rho_2),\,\half d(L_1,L_2) &
d(\nu_1,\nu_2) &\le \max d_{e_2}(\rho_1,\rho_2),\,\half d(L_1,L_2),
\end{align*}
where
\begin{align*}
\mu_1 &= \tau^{n}(L_1)(e_1,\rho_1) &
\nu_1 &= \tau^{n}(L_1)(e_2,\rho_1) &
\mu_2 &= \tau^{n}(L_2)(e_1,\rho_2) &
\nu_2 &= \tau^{n}(L_2)(e_2,\rho_2).
\end{align*}

Note that $U(e_1\cmp e_2) = U(e_2)$ and $G(e_1\cmp e_2) = G(e_2)$ if $\skp$ is unguarded in $e_1$, and $U(e_1\cmp e_2) = \emptyset$ and $G(e_1\cmp e_2) = U(e_2)\cup G(e_2)$ if $\skp$ is guarded in $e_1$. In the former case, $d_{e_2}(\rho_1,\rho_2) = d_{e_1\cmp e_2}(\rho_1,\rho_2)$, and in the latter,
\begin{align*}
\half d_{e_2}(\rho_1,\rho_2)
&= \half (\max d_{U(e_2)}(\rho_1,\rho_2),\,\half d_{G(e_2)}(\rho_1,\rho_2))\\
&\le \half (\max d_{U(e_2)}(\rho_1,\rho_2),\,d_{G(e_2)}(\rho_1,\rho_2))\\
&= \half d_{U(e_2)\cup G(e_2)}(\rho_1,\rho_2)\\
&= \max d_{U(e_1\cmp e_2)}(\rho_1,\rho_2),\,\half d_{G(e_1\cmp e_2)}(\rho_1,\rho_2))\\
&= d_{e_1\cmp e_2}(\rho_1,\rho_2).
\end{align*}

If $\skp$ is unguarded in $e_1$, then using Lemma \ref{lem:ultrametricB}(iv),
\begin{align*}
d(\tau^{n+1}(L_1)(e_1 \cmp e_2,\rho_1),\tau^{n+1}(L_2)(e_1 \cmp e_2,\rho_2))
&= d(\mu_1\bind -\cdot\nu_1),\mu_2\bind -\cdot\nu_2)\\
&\le \max d(\mu_1,\mu_2),\ d(\nu_1,\nu_2)\\
&\le \max d_{e_1}(\rho_1,\rho_2),\,d_{e_2}(\rho_1,\rho_2),\,\half d(L_1,L_2)\\
&= \max d_{e_1\cmp e_2}(\rho_1,\rho_2),\,\half d(L_1,L_2).
\end{align*}

If $\skp$ is guarded in $e_1$, we still have $d(\nu_1,\nu_2)\le d_{e_2}(\rho_1,\rho_2)$ by the induction hypothesis. Moreover, $\mu_1(\set{\alpha}{\alpha(\eps)>0}) = 0$, or equivalently, $\mu_1([0]_0) = 1$; that is, with probability 1, a sample from $\mu_1$ does not contain $\eps$ in its support. Thus if $\nu_1\equiv_n\nu_2$, then by Lemma \ref{lem:bind}(ii),
\begin{align*}
\mu_1\bind-\cdot\nu_1\equiv_{n+1}\mu_1\bind-\cdot\nu_2, 
\end{align*}
so
\begin{align*}
d(\mu_1\bind-\cdot\nu_1,\mu_1\bind-\cdot\nu_2) &\le \half d(\nu_1,\nu_2)
\le \half(\max d_{e_2}(\rho_1,\rho_2),\,\half d(L_1,L_2)).
\end{align*}
In this case,
\begin{align*}
& d(\tau^{n+1}(L_1)(e_1 \cmp e_2,\rho_1),\tau^{n+1}(L_2)(e_1 \cmp e_2,\rho_2))\\
&= d(\mu_1\bind -\cdot\nu_1,\mu_2\bind -\cdot\nu_2)\\
&\le \max d(\mu_1\bind -\cdot\nu_1,\mu_1\bind -\cdot\nu_2),\ d(\mu_1\bind -\cdot\nu_2,\mu_2\bind -\cdot\nu_2)\\
&\le \max d(\mu_1,\mu_2),\,\half d(\nu_1,\nu_2)\\
&\le \max d_{e_1}(\rho_1,\rho_2),\,\half d_{e_2}(\rho_1,\rho_2),\,\half d(L_1,L_2)\\
&\le \max d_{e_1\cmp e_2}(\rho_1,\rho_2),\,\half d(L_1,L_2).
\end{align*}

For a fixpoint expression $\fix xe$, by choosing $n$ sufficiently large, the induction hypothesis gives
\begin{align*}
d(\tau^n(L_1)(e,\rho_1),\tau^n(L_2)(e,\rho_2)) &\le \max d_e(\rho_1,\rho_2),\,\half d(L_1,L_2)
\end{align*}
for any $\rho_1,\rho_2$. Let
\begin{align*}
\rho_1' &= \rho_1[\tau^n(L_1)(\fix xe,\rho_1)/x] &
\rho_2' &= \rho_2[\tau^n(L_2)(\fix xe,\rho_2)/x]).
\end{align*}
Since $U(e)=U(\fix xe)$ and $G(e)=G(\fix xe)\cup\{x\}$,
\begin{align*}
& d(\tau^{n+1}(L_1)(e,\rho_1),\tau^{n+1}(L_2)(e,\rho_2))\\
&= d(\tau^n(L_1)(e,\rho_1'),\tau^n(L_2)(e,\rho_2'))\\
&\le \max d_e(\rho_1',\rho_2'),\,\half d(L_1,L_2)\\
&= \max d_{U(e)}(\rho_1',\rho_2'),\ \half d_{G(e)}(\rho_1',\rho_2'),\,\half d(L_1,L_2)\\
&= \max d_{U(\fix xe)}(\rho_1,\rho_2),\ \half d_{G(\fix xe)}(\rho_1,\rho_2),\ \half d(\rho_1'(x),\rho_2'(x)),\,\half d(L_1,L_2)\\
&= \max d_{\fix xe}(\rho_1,\rho_2),\,\half d(L_1,L_2),\ \half d(\tau^n(L_1)(\fix xe,\rho_1),\tau^n(L_2)(\fix xe,\rho_2)).
\end{align*}
We now use this to show that
\begin{align*}
& d(\tau^{n+k}(L_1)(\fix xe,\rho_1),\tau^{n+k}(L_2)(\fix xe,\rho_2))\\
&\le \max d_{\fix xe}(\rho_1,\rho_2),\,\half d(L_1,L_2),\ 2^{-k}d(\tau^n(L_1)(\fix xe,\rho_1),\tau^n(L_2)(\fix xe,\rho_2)).
\end{align*}
If $\rho_1=\rho_2$, there is nothing to prove. Otherwise, we proceed by induction.
The claim is true for $k=0$ by inspection, and
\begin{align*}
& d(\tau^{n+k+1}(L_1)(\fix xe,\rho_1),\tau^{n+k+1}(L_2)(\fix xe,\rho_2))\\
&\le \max d_{\fix xe}(\rho_1,\rho_2),\,\half d(L_1,L_2),\ \half d(\tau^{n+k}(L_1)(\fix xe,\rho_1),\tau^{n+k}(L_2)(\fix xe,\rho_2))\\
&\le \max d_{\fix xe}(\rho_1,\rho_2),\,\half d(L_1,L_2),\\
&\qquad \half(\max d_{\fix xe}(\rho_1,\rho_2),\,\half d(L_1,L_2),\ 2^{-k}d(\tau^n(L_1)(\fix xe,\rho_1),\tau^n(L_2)(\fix xe,\rho_2)))\\
&= \max d_{\fix xe}(\rho_1,\rho_2),\,\half d(L_1,L_2),\ 2^{-(k+1)}d(\tau^n(L)(\fix xe,\rho_1),\tau^n(L)(\fix xe,\rho_2)).
\end{align*}
Choosing $k\ge 1-\log_2 d(L_1,L_2)$ so that $2^{-k}\le \half d(L_1,L_2)$ and using the fact that all distances are bounded by 1, we have
\begin{align*}
& d(\tau^{n+k}(L_1)(\fix xe,\rho_1),\tau^{n+k}(L_2)(\fix xe,\rho_2))\\
&\le \max d_{\fix xe}(\rho_1,\rho_2),\,\half d(L_1,L_2),\ 2^{-k}d(\tau^n(L_1)(\fix xe,\rho_1),\tau^n(L_2)(\fix xe,\rho_2))\\
&= d_{\fix xe}(\rho_1,\rho_2),\,\half d(L_1,L_2).
\tag*\qedhere
\end{align*}
\end{proof}

\begin{lemmarep}
\label{lem:expcontractive}
For any $e\in\Exp$, the map $\tau$ of \eqref{eq:exptaudef} restricted to subterms of $e$ is eventually contractive with (eventual) constant of contraction $1/2$.
\end{lemmarep}

\begin{proof}
Let $A$ be the set of subterms of a given term.
By Lemma \ref{lem:expcontractiveA}, for any $e\in A$, for sufficiently large $n$ and any $\rho:\Var\to\DNS$,
\begin{align*}
d(\tau^n(L_1)(e,\rho),\tau^n(L_2)(e,\rho)) &\le \max d_e(\rho,\rho),\,\half d(L_1,L_2) = \half d(L_1,L_2).
\end{align*}
As $e$ and $\rho$ were arbitrary, $d(\tau^n(L_1),\tau^n(L_2)) \le \half d(L_1,L_2)$.
\end{proof}

\begin{theoremrep}
\label{thm:expsemwelldefined}
The semantic map $\sem-$ on expressions is well defined.
\end{theoremrep}
\begin{proof}
The proof is the same as the proof of Theorem \ref{thm:autsemwelldefined}.
\end{proof}

Contractive maps also allow proofs by coinduction in which one can assume a coinduction hypothesis equivalent to the proposition to be proved, as long as progress is made, as explained in \cite{KS16a}. These arguments are sometimes called cyclic proof systems. The proof of the following lemma is an example. The lemma relates substitution to rebinding and is a well-known phenomenon in logical systems; see for example \cite[Lemma 5.15(i), p.~89]{Barendregt84}.

\begin{lemmarep}
\label{lem:substrebind}
$\sem{e[d/x]}\rho = \sem e\rho[\sem d\rho/x]$.
\end{lemmarep}
\begin{proof}
By coinduction. The base cases are all straightforward:
\begin{align*}
\sem{x[d/x]}\rho &= \sem d\rho = \sem x\rho[\sem d\rho/x]\\
\sem{\skp[d/x]}\rho &= \sem\skp\rho = \dirac{\mset\eps} = \sem\skp\rho[\sem d\rho/x]\\
\sem{\fail[d/x]}\rho &= \sem\fail\rho = \dirac{\mset{}} = \sem\fail\rho[\sem d\rho/x]\\
\sem{a[d/x]}\rho &= \sem a \rho = \dirac{\mset a} = \sem a \rho[\sem d\rho/x],\ a\in\Sigma.
\end{align*}
The cases $\oplus$ and $\amp$ are also straightforward:
\begin{align*}
\sem{(e_1 \oplus_r e_2)[d/x]}\rho &= \sem{e_1[d/x] \oplus_r e_2[d/x]}\rho = \sem{e_1[d/x]} \oplus_r \sem{e_2[d/x]}\rho\\
&= \sem{e_1}\rho[\sem d\rho/x] \oplus_r \sem{e_2}\rho[\sem d\rho/x] = \sem{e_1 \oplus_r e_2}\rho[\sem d\rho/x]
\end{align*}
and similarly for $\amp$. For compositions, cognizant of the restriction the left operand must be closed (may not contain any free variables),
\begin{align*}
\sem{(e_1 \cmp e_2)[d/x]}\rho &= \sem{e_1 \cmp (e_2[d/x])}\rho = \sem{e_1} \bind -\cdot\sem{e_2[d/x]}\rho\\
&= \sem{e_1} \bind -\cdot\sem{e_2}\rho[\sem d\rho/x] = \sem{e_1 \cmp e_2}\rho[\sem d\rho/x]
\end{align*}
For fixpoint expressions, we have two cases.
\begin{align*}
\sem{(\fix xe)[d/x]}\rho &= \sem{\fix xe}\rho = \sem{\fix xe}\rho[\sem d\rho/x]
\end{align*}
For $y\ne x$, we assume without loss of generality that $y$ is not free in $d$.
\begin{align*}
& \sem{(\fix ye)[d/x]}\rho\\
&= \sem{\fix y{(e[d/x])}}\rho && \text{definition of substitution}\\
&= \sem{e[d/x]}\rho[\sem{\fix y{(e[d/x])}}\rho/y] && \text{semantics of $\mathsf{fix}$}\\
&= \sem e\rho[\sem{\fix y{(e[d/x])}}\rho/y][\sem d\rho[\sem{\fix y{(e[d/x])}}\rho/y]/x] && \text{coinduction hypothesis}\\
&= \sem e\rho[\sem{\fix y{(e[d/x])}}\rho/y][\sem d\rho/x] && \text{since $y$ is not free in $d$}\\
&= \sem e\rho[\sem d\rho/x][\sem{\fix y{(e[d/x])}}\rho/y] && \text{switch order of rebinding}\\
&= \sem e\rho[\sem d\rho/x][\sem{(\fix ye)[d/x]}\rho/y] && \text{definition of substitution}\\
&= \sem e\rho[\sem d\rho/x][\sem{\fix ye}\rho[\sem d\rho/x]/y] && \text{coinduction hypothesis}\\
&= \sem{\fix ye}\rho[\sem d\rho/x] && \text{semantics of $\mathsf{fix}$.}
\tag*\qedhere
\end{align*}
\end{proof}

\section{A Kleene Theorem}
\label{sec:kleene}

In this section, we prove a Kleene theorem showing that expressions and automata are equivalent in expressive power. We first describe a third formalism that is equivalent to both expressions and automata involving systems of affine linear equations.

\subsection{Systems of Equations}
\label{sec:systemequations}

A \emph{system of (affine linear) equations} is a finite collection of equations of the form $x = e$, where $x$ is a variable and $e$ is an expression, with certain restrictions listed below.

An occurrence of a variable $x$ in $e$ is \emph{guarded} in $e$ if it is in a subexpression of $e$ of the form $a \cmp e'$, otherwise it is \emph{unguarded}. By $\alpha$-conversion if necessary, we can assume that all bound variables are distinct and different from all free variables in the system. 

We assume that systems of equations satisfy the following restrictions:
\begin{itemize}
  \item No variable may appear on the left-hand side of more than one equation. 
  \item Compositions may occur only in the form $a \cmp e$ for $a \in \Sigma$.
  \item There must be no sequence of equations $x_i = e_i$, $1\le i\le n$, in which $x_{i + 1}$ occurs unguarded in $e_i$, $1\le i\le n-1$, and $x_1$ occurs unguarded in $e_n$. This is the productivity assumption again.
\end{itemize}

Similar to automata and expressions, the semantics of systems of equations is defined coinductively, relative to an environment $\rho:\Var\to\DNS$ to interpret free variables.
\begin{itemize}
\item 
$\semr{\skp} = \dirac{\mset\eps}$
\item 
$\semr{\fail} = \dirac{\mset{}}$
\item 
$\semr{a \cmp e} = a\cdot\semr{e}$
\item 
$\semr{e_1\oplus_r e_2} = r\sem{e_1}\rho + (1-r)\semr{e_2}$
\item 
$\semr{e_1\amp e_2} = \semr{e_1} \amp \semr{e_2}$
\item
$\semr{\fix xe} = \semr{e}[\semr{\fix xe}/x]$
\item
If $x$ occurs on the left of an equation $x = e$, then $\semr{x} = \semr{e}$
\item
If $x$ does not occur on the left of an equation, then $\semr{x} = \rho(x)$.
\end{itemize}

\begin{lemmarep}
\label{lem:semsystemwelldefined}
Given a fixed system of equations, the semantic map $\sem{-}: \Exp \to (\Var \to \DNS) \to \DNS$ is well defined.
\end{lemmarep}
\begin{proof}
The proof is the same as for Theorems \ref{thm:autsemwelldefined} and \ref{thm:expsemwelldefined}, using the version of $\semr-$ defined for systems of equations and a corresponding contractive map.
Here, the labeling functions will take the form $L : E \to (F \to \DNS) \to \DNS$.
Like before, we can establish a metric $d(L_1,L_2)$ on labelings as follows:
\begin{align*}
d(L_1,L_2) &= \begin{cases}
2^{-k}, & \text{if $k$ is minimum such that there exist $e\in E$, $m\in \NS$, $n\in\naturals$,}\\
& \text{and $\rho : F \to \DNS$ with $k \ge N(e,n)$ and $L_1(e)\rho([m]_n)\ne L_2(e)\rho([m]_n)$,}\\
0, & \text{if no such $k$ exists.}
\end{cases}
\end{align*}
The proof that this is a complete ultrametric follows the same structure as that of Lemma \ref{lem:ultrametric}, with the version of $d$ for systems of equations replacing the one defined for automata.

Semantics will again be the fixpoint of a contractive map $\tau: (\Exp \to (F \to \DNS) \to \DNS) \to (\Exp \to (F \to \DNS) \to \DNS)$. Let $L: \Exp \to (F \to \DNS) \to \DNS$ and define $\tau(L)$:
\begin{itemize}
\item 
If $e=\skp$, then $\tau(L)(e)\rho = \dirac{\mset\eps}$.
\item 
If $e=\fail$, then $\tau(L)(e)\rho = \dirac{\mset{}}$.
\item 
If $e=p \cmp f$, then $\tau(L)(e)\rho = p\cdot L(f)\rho$.
\item 
If $e=\bigoplus_i r_if_i$, then $\tau(L)(e)\rho = \sum_i r_i L(f_i)\rho$.
\item 
If $e=\bigamp_{i=1}^k f_i$, then $\tau(L)(e)\rho = {\amp}(ML(\mset{\seq f1k})\rho)$.
\item 
If $e = x \in F$, then $\tau(L)(e)\rho = \rho(x)$.
\item
If $e = x_i \in X \setminus F$ with equation $x_i = e_i$, then $\tau(L)(e)\rho = L(e_i)\rho$.
\end{itemize}
The proof that this is a contraction is largely the same as that of Lemma \ref{lem:autcontractive}, with additional cases to handle bound variables and concatenation.

Suppose that $d(L_1,L_2) < 2^{-k}$. We wish to show that 
\begin{align}
\forall e\ \forall n\ \forall m\ \forall \rho\ \ k+1\ge N(e,n) \Imp \tau(L_1)(e)\rho([m]_n)=\tau(L_2)(e)\rho([m]_n).\label{eq:contractiveBsystem}
\end{align}
Let $e, n, m, \rho$ be arbitrary and suppose that $k + 1 \geq N(e,n)$. 

We treat the case where $e = x_i \in X \setminus F$ is a bound variable.
Note that $k + 1 \geq N(x_i,n)$ implies $k \geq N(e_i,n)$ so $L_1(e_i)\rho([m]_n) = L_2(e_i)\rho([m]_n)$.
\begin{align*}
  \tau(L_1)(e)\rho([m]_n) = L_1(e_i)\rho([m]_n) = L_2(e_i)\rho([m]_n) = \tau(L_2)(e)\rho([m]_n)
\end{align*}

Now consider the case where $e = p \cmp f$.

Note that $(p \cdot D)([m]_n) = D( (p \cdot -)^{-1}([m]_n))$. Thus, since $\tau(L)(e)\rho = p \cdot L(f)\rho$ it suffices to show that $L_1(f)\rho((p \cdot -)^{-1}([m]_n)) = L_2(f)\rho((p \cdot -)^{-1}([m]_n))$. From Lemma \ref{lem:sumprod}, we have three cases to consider:
\begin{itemize}
  \item If $n = 0$ and $m(\eps) = 0$, then $(p \cdot -)^{-1}([m]_n) = \NS$. Then,
  \begin{align*}
    L_1(f)\rho((p \cdot -)^{-1}([m]_n)) = L_1(f)\rho(\NS) = 1 = L_2(f)\rho(\NS) = L_2(f)\rho((p \cdot -)^{-1}([m]_n))
  \end{align*}
  
  \item If $n \geq 1$ and $m \equiv_n p \cdot m'$, then $(p \cdot -)^{-1}([m]_n) = [m']_{n-1}$. From $k + 1 \geq N(e,n)$, we have that $k \geq N(f,n-1)$, so $L_1(f)\rho ([m']_{n-1}) = L_2(f)\rho([m']_{n-1})$. Then,
  \begin{align*}
    L_1(f)\rho((p \cdot -)^{-1}([m]_{n})) = L_1(f)\rho([m']_{n-1}) = L_2(f)\rho([m']_{n-1}) = L_2(f)\rho((p \cdot -)^{-1}([m]_{n}))
  \end{align*}

  \item Otherwise, $(p \cdot -)^{-1}([m]_n) = \emptyset$. Then,
  \begin{align*}
    L_1(f)\rho((p \cdot -)^{-1}([m]_n)) = L_1(f)\rho(\emptyset) = 0 = L_2(f)\rho(\emptyset) = L_2(f)\rho((p \cdot -)^{-1}([m]_n))
  \end{align*}
  
\end{itemize}
\end{proof}

Every expression and every automaton can be converted to an equivalent system of equations. For automata, we can introduce a variable for each state and an equation describing the transitions from that state. The systems resulting from this transformation contain no occurrence of the $\mathsf{fix}$ operator. We will show in \S\ref{sec:automatontoexpression} that the system of equations arising from any automaton has a solution in the space of expressions.

For expressions, the translation is less straightforward, chiefly due to the non-local effect of sequential composition ($\cmp$). We will explain how to handle this in \S\ref{sec:expressiontoautomaton}.

\subsection{Expressions to Automata}
\label{sec:expressiontoautomaton}

We can now use the results of \S\ref{sec:systemequations} to convert expressions to equivalent automata. This is one direction of our Kleene theorem.

The main difficulty is the sequential composition operator ($\cmp$). The issue here is that the parse tree of the expression is not the same as the computation tree of the corresponding automaton. The two trees (parse and computation) are the same for the other operators: an agent at a state corresponding to the expression $e_1 \oplus_{1/2} e_2$ would flip a coin and move down to one of $e_1$ or $e_2$, and an agent at a state corresponding to the expression $e_1\amp e_2$ would fork and send an agent to each of $e_1$ or $e_2$. However, for the expression $e_1\cmp e_2$, one should think of $e_2$ as sitting \emph{below} $e_1$ in the corresponding computation tree. That is, when an agent reaches a leaf of $e_1$, say $\skp$ or $a$, it would move down to the root of the computation tree corresponding to $e_2$ and continue from there. The substitution operator $[e]$ defined below does exactly this transformation, essentially transforming the parse tree to a computation tree. Lemma \ref{lem:elimcmp} below proves its correctness.

We now show how to eliminate sequential composition. By this, we mean that we can reduce expressions to a form in which the sequential composition operator appears only in the form $a\cmp e$, where $a$ is a primitive letter, as required by restrictions on systems of equations. Let us define the following postfix syntactic substitution operator:
\begin{align*}
[e] = [e/\skp,(a\cmp e)/a \mid a\in\Sigma].
\end{align*}
Applying this operator to an expression $e_1$ simultaneously substitutes $e$ for $\skp$ and $a\cmp e$ for $a$ at all occurrences of $\skp$ and $a$ appearing in a terminal position in $e_1$. Intuitively, this is meant to capture the idea that any $a$ and $\skp$ performed last in a computation should be followed by the continuation $e$. The formal definition is inductive.
\begin{gather*}
\skp[e] = e
\qquad\qquad \fail[e] = \fail
\qquad\qquad a[e] = a\cmp e
\qquad\qquad x[e] = x\\[1ex]
(e_1\oplus_r e_2)[e] = (e_1[e])\oplus_r(e_2[e])
\qquad\qquad (e_1\cmp e_2)[e] = e_1\cmp(e_2[e])\\[1ex]
(e_1\amp e_2)[e] = (e_1[e])\amp(e_2[e])
\qquad\qquad (\fix x{e_1})[e] = \fix x{(e_1[e])}
\end{gather*}
In the clause for the fixpoint expression, we assume without loss of generality that $x$ has no free occurrence in $e$. This can be enforced by $\alpha$-conversion if necessary.

\begin{lemmarep}
\label{lem:elimcmp}
$\semr{e_1\cmp e_2} = \semr{e_1[e_2]}$.
\end{lemmarep}
\begin{proof}
By coinduction on the structure of $e_1$. Coinduction, as opposed to induction, is needed to deal with the case $\fix xe$, since unwinding the definition results in a larger term.

By definition,
\begin{align*}
& \semr{\skp\cmp e} = \semr{e} = \semr{\skp[e]}
&& \semr{a\cmp e} = \semr{a[e]}
&& \semr{\fail\cmp e} = \semr{\fail} = \semr{\fail[e]}.
\end{align*}
The case $\semr{x\cmp e}$ cannot occur by the restriction that variables must occur in tail position.

For the remaining cases, by definition, $\semr{e_1\cmp e_2} = \semr{e_1}\bind -\cdot\semr{e_2}$, and
\begin{align*}
\mu \bind f &= \mu\circ({\amp}\circ Mf)^{-1} = \mu\circ(Mf)^{-1}\circ{\amp}^{-1},
\end{align*}
so we must show that
\begin{align*}
\semr{e_1[e_2]} &= \semr{e_1}\circ(M(-\cdot\semr{e_2}))^{-1}\circ{\amp}^{-1}.
\end{align*}

For the case $e_1\oplus_r e_2$, we first show that for any $g$,
\begin{align}
(\mu\oplus_r\nu)\circ g^{-1}
&= (\mu\circ g^{-1}) \oplus_r (\nu\circ g^{-1}).\label{eq:plusr1}
\end{align}
For any $A$,
\begin{align*}
((\mu\oplus_r\nu)\circ g^{-1})(A)
&= (\mu\oplus_r\nu)(g^{-1}(A))
= (r\mu + (1-r)\nu)(g^{-1}(A))\\
&= r\mu(g^{-1}(A)) + (1-r)\nu(g^{-1}(A))
= r(\mu\circ g^{-1})(A) + (1-r)(\nu\circ g^{-1})(A)\\
&= (r(\mu\circ g^{-1}) + (1-r)(\nu\circ g^{-1}))(A)
= ((\mu\circ g^{-1}) \oplus_r (\nu\circ g^{-1}))(A).
\end{align*}
It follows that for any $f$,
\begin{align}
(\mu\oplus_r\nu)\bind f
&= (\mu\oplus_r\nu)\circ({\amp}\circ Mf)^{-1}\nonumber\\
&= (\mu\circ({\amp}\circ Mf)^{-1}) \oplus_r (\nu\circ({\amp}\circ Mf)^{-1}) && \text{by \eqref{eq:plusr1}}\nonumber\\
&= (\mu\bind f) \oplus_r (\nu\bind f)\label{eq:plusr}
\end{align}
Then
\begin{align*}
\semr{e_1\oplus_r e_2}\bind -\cdot\semr{e}
&= (\semr{e_1} \oplus_r\semr{e_2})\bind -\cdot\semr{e}\\
&= (\semr{e_1}\bind -\cdot\semr{e}) \oplus_r (\semr{e_2}\bind -\cdot\semr{e}) && \text{by \eqref{eq:plusr}}\\
&= \semr{e_1[e]}\oplus_r\semr{e_2[e]} && \text{coinductive hypothesis}\\
&= \semr{(e_1[e])\oplus_r (e_2[e])}\\
&= \semr{(e_1\oplus_r e_2)[e]}.
\end{align*}

For the case $e_1\amp e_2$, we first show that if $g$ is additive, that is, if $g\circ (+) = (+)\circ(g \times g)$, then
\begin{align}
(\mu\oplus_r\nu)\circ g^{-1}
&= (\mu\circ g^{-1}) \oplus_r (\nu\circ g^{-1}).\label{eq:amp}
\end{align}
\begin{align*}
(\mu\amp\nu)\circ g^{-1}
&= (\mu\otimes\nu)\circ(+)^{-1}\circ g^{-1}
= (\mu\otimes\nu)\circ(g\circ (+))^{-1}\\
&= (\mu\otimes\nu)\circ((+)\circ(g \times g))^{-1} && \text{by \eqref{eq:amp}}\\
&= (\mu\otimes\nu)\circ(g \times g)^{-1}\circ(+)^{-1}
= ((\mu\circ g^{-1})\otimes(\nu\circ g^{-1}))\circ(+)^{-1}\\
&= ((\mu\circ g^{-1})\amp(\nu\circ g^{-1})).
\end{align*}
Both $\amp$ an $Mf$ are additive:
\begin{align*}
& Mf(m+\ell) = Mf(m) + Mf(\ell) && {\amp}(m+\ell) = {\amp}(m) + {\amp}(\ell),
\end{align*}
therefore so is their composition ${\amp}\circ Mf$, so
\begin{align*}
{\amp}\circ Mf\circ (+)
&= (+)\circ(({\amp}\circ Mf) \times ({\amp}\circ Mf)).
\end{align*}
Then
\begin{align}
(\mu\amp\nu)\bind f
&= (\mu\amp\nu)\circ({\amp}\circ Mf)^{-1}\nonumber\\
&= (\mu\circ({\amp}\circ Mf)^{-1})\amp(\nu\circ({\amp}\circ Mf)^{-1}) && \text{by \eqref{eq:amp}}\nonumber\\
&= (\mu\bind f)\amp(\nu\bind f),\label{eq:amp1}
\end{align}
therefore
\begin{align*}
\semr{e_1\amp e_2}\bind -\cdot\semr{e}
&= (\semr{e_1}\amp\semr{e_2})\bind -\cdot\semr{e}\\
&= (\semr{e_1}\bind -\cdot\semr{e})\amp(\semr{e_2}\bind -\cdot\semr{e}) && \text{by \eqref{eq:amp1}}\\
&= \semr{(e_1[e])}\amp\semr{(e_2[e])} && \text{coinductive hypothesis}\\
&= \semr{(e_1[e])\amp (e_2[e])}\\
&= \semr{(e_1\amp e_2)[e]}.
\end{align*}

For the case $e_1\cmp e_2$,
\begin{align*}
\semr{(e_1\cmp e_2)\cmp e}
&= \semr{e_1\cmp(e_2\cmp e)} && \text{associativity}\\
&= \semr{e_1} \bind -\cdot\semr{e_2\cmp e]} && \text{definition of $\cmp$}\\
&= \semr{e_1} \bind -\cdot\semr{e_2[e]} && \text{coinductive hypothesis}\\
&= \semr{e_1\cmp(e_2[e])} && \text{definition of $\cmp$}\\
&= \semr{(e_1\cmp e_2)[e]} && \text{definition of $[e]$.}
\end{align*}

For the final case $\fix x{e_1}$, we wish to show
\begin{align*}
\semr{(\fix x{e_1});e} &= \semr{(\fix x{e_1})[e]}
\end{align*}
As mentioned, by $\alpha$-conversion if necessary, we can assume without loss of generality that there is no free occurrence of $x$ in $e$. Due to this assumption, we have
\begin{align}
e_1[e_2/x][e] &= e_1[e][e_2[e]/x].\label{eq:fix1}
\end{align}
Then
\begin{align*}
\semr{(\fix x{e_1})[e]}
&= \semr{\fix x{(e_1[e])}} && \text{definition of $[e]$}\\
&= \semr{e_1[e][\fix x{(e_1[e])}/x]}\\
&= \semr{e_1[e][(\fix x{e_1})[e]/x]} && \text{definition of $[e]$}\\
&= \semr{e_1[\fix x{e_1}/x][e]} && \text{by \eqref{eq:fix1}}\\
&= \semr{e_1[\fix x{e_1}/x]\cmp e} && \text{coinductive hypothesis}\\
&= \semr{(\fix x{e_1})\cmp e}.
\tag*\qedhere
\end{align*}
\end{proof}

Now we show how to convert an expression to an equivalent automaton. We first rename all bound variables as necessary to avoid duplication. Then we apply Lemma \ref{lem:elimcmp} to transform the expression to an equivalent one in which all compositions are of the basic form $a\cmp d$ with $a\in\Sigma$. Let $e$ be this new expression. The states of our automaton will be the subexpressions of $e$. Let us write $\auto(e)$ for the state of the automaton corresponding to the expression $e$. Actually, $\auto$ is the identity function, but the difference is that $\sem e\rho$ will refer to the semantics of expressions as given in \S\ref{sec:semanticsofexpressions}, whereas $\sem{\auto(e)}$ will refer to the semantics of automata as given in \S\ref{sec:semanticsofautomata}. The labels and transitions of the automaton are given Table \ref{tab:expressionstoautomata}.
\begin{table}
\caption{Converting an expression to an automaton}
\label{tab:expressionstoautomata}
\begin{tabular}{l@{\hspace{1cm}}l@{\hspace{1cm}}l}
  \toprule
   Expression $e$ & $\ell(\auto(e))$ & $\del(\auto(e))$\\
  \midrule
   $\skp$ & $\skp$ & $-$\\
   $\fail$ & $\fail$ & $-$\\
   $a\in\Sigma$ & $a$ & $h(\skp)$\\
   $a\cmp d$, $a\in\Sigma$ & $a$ & $\auto(d)$\\
   $e_1 \amp e_2$ & $\amp$ & $\mset{\auto(e_1),\auto(e_2)}$\\
   $e_1 \oplus_r e_2$ & $\oplus$ & $r\auto(e_1)+(1-r)\auto(e_2)$\\
   $\fix xd$ & $\amp$ & $\mset{\auto(d)}$\\
   $x$ & $\amp$ & $\mset{\auto(\fix xd)}$\\
  \bottomrule 
\end{tabular}
\end{table}

Here are some observations about this construction:
\begin{itemize}
\item
In the clause for $a\in\Sigma$, we treat $a$ as we would $a\cmp\skp$. We may have to introduce the term $\skp$ if it does not already occur in $e$.
\item
In the clauses for $\fix xd$ and $x$, we have taken $\ell(\auto(e))={\amp}$ and $\del(\auto(e))$ a singleton multiset, but we could have defined $\ell(e)=\oplus$ and $\del(\auto(e))$ a Dirac measure on the successor. In both cases there is a single successor and the effect is the same.
\item
In the clause for $x$, the transition function $\del(\auto(x))$ takes $\auto(x)$ to the image under $\auto$ of the binding occurrence $\fix xd$ of $x$, where $x$ is a free variable of $d$. Any cycle created by this back edge must visit an action state because of the productivity assumption for expressions, so the automaton also satisfies the productivity assumption.
\end{itemize}

\begin{lemmarep}
\label{lem:expressiontoautomaton}
Let $e$ be a subexpression of a closed expression. Let $\rho:\Var\to\DNS$ be an environment such that $\rho(x)=\sem{\fix xd}\rho$ for all free variables $x$ of $e$, where $\fix xd$ is the binding occurrence of $x$; that is, $\rho=\rho[\sem{\fix xd}\rho/x]$. Then $\sem e\rho = \sem{\auto(e)}$. In particular, $\sem e = \sem{\auto(e)}$ for all closed terms $e$.
\end{lemmarep}
\begin{proof}
It may not be immediately clear that there exists an environment $\rho$ satisfying the preconditions of the lemma. Let $\fix{x_i}{e_i}$, $1\le i\le n$, be the sequence of fixpoint expressions in whose scope $e$ occurs in order from outermost to innermost. Given an arbitrary initial environment $\rho:\Var\to\DNS$, let
\begin{align*}
\rho_0 &= \rho & \rho_i &= \rho_{i-1}[\sem{\fix{x_i}{e_i}}\rho_{i-1}/x_i],\ \ 1\le i\le n.
\end{align*}
Then $\rho_n(x_i) = \rho_i(x_i) = \sem{\fix{x_i}{e_i}}\rho_{i-1}$, since bound variables are distinct. But $\sem{\fix{x_i}{e_i}}\rho_{i-1} = \sem{\fix{x_i}{e_i}}\rho_n$, because $\rho_{i-1}$ and $\rho_n$ agree on all free variables of $\fix{x_i}{e_i}$, namely $\{x_1,\ldots,x_{i-1}\}$. Thus $\rho_n(x_i) = \sem{\fix{x_i}{e_i}}\rho_n$, that is, $\rho_n = \rho_n[\sem{\fix{x_i}{e_i}}\rho_n/x_i]$, $1\le i\le n$.

The proof is by coinduction. In most cases, this amounts to comparing the coinductive definitions of the semantics of each operator. For each of the cases $e$ below, suppose $\rho(x)=\sem{\fix xd}\rho$ for all free variables $x$ of $e$; that is, for each free variable $x$ of $e$, $\rho[\sem{\fix xd}\rho/x] = \rho$. For the coinductive cases, there is an appeal to the coinductive hypothesis.
\begin{align*}
\sem{\auto(\skp)} &= \dirac{\mset\eps} = \sem\skp\rho\\
\sem{\auto(\fail)} &= \dirac{\mset{}} = \sem\fail\rho\\
\sem{\auto(a)} &= a\cdot\sem{\auto(\skp)} = a\cdot\dirac{\mset\eps} = \dirac{\mset a} = \sem a\rho\\
\sem{\auto(a\cmp d)} &= a\cdot\sem{\auto(d)} = a\cdot\sem d\rho = \sem{a\cmp d}\rho\\
\sem{\auto(x)} &= \sem{{\amp}\mset{\auto(\fix xd)}} = \sem{\auto(\fix xd)} = \sem{\fix xd}\rho = \rho(x) = \sem{x}\rho\\
\sem{\auto(\fix xd)} &= \sem{{\amp}\mset{\auto(d)}} = \sem{\auto(d)} = \sem{d}\rho[\sem{\fix xd}\rho/x] = \sem{\fix xd}\rho\\
\sem{\auto(e_1 \amp e_2)} &= \sem{{\amp}\mset{\auto(e_1),\auto(e_2)}} = {\amp}\mset{\sem{\auto(e_1)},\sem{\auto(e_2)}} = {\amp}\mset{\sem{e_1}\rho,\sem{e_2}\rho}\\
&= \sem{e_1}\rho \amp \sem{e_2}\rho = \sem{e_1\amp e_2}\rho\\
\sem{\auto(e_1 \oplus_r e_2)} &= \sem{\auto(e_1) \oplus_r\auto(e_2)} = r\sem{\auto(e_1)}+(1-r)\sem{\auto(e_2)}\\
&= r\sem{e_1}\rho+(1-r)\sem{e_2}\rho = \sem{e_1}\rho\oplus_r\sem{e_2}\rho = \sem{e_1\oplus_re_2}\rho.
\tag*\qedhere
\end{align*}
\end{proof}

\begin{theoremrep}
\label{thm:expressiontoautomaton}
For every closed expression $e\in\Exp$, there exists an automaton with equivalent semantics. 
\end{theoremrep}
\begin{proof}
This follows immediately from Lemma \ref{lem:expressiontoautomaton}.
\end{proof}

\subsection{Automata to Expressions}
\label{sec:automatontoexpression}

As a first step in converting automata to expressions, we observe that automata are essentially systems of equations. Given an automaton $(S,\ell,\del)$, we can regard the states $S$ as variables and the transition structure as equations, according to the following table:
\begin{align*}
  \begin{tabular}{ccc}
  \toprule
    $\ell(s)$ & $\del(s)$ & \text{corresponding equation}\\
  \midrule
    $\skp$ & $-$ & $s = \skp$\\
    $\fail$ & $-$ & $s = \fail$\\
    $a\in\Sigma$ & $t$ & $s = a\cmp t$\\
    $\oplus$ & $\sum_{i=1}^k r_it_i$ & $s = \oplus_{i=1}^k r_it_i$\\
    $\amp$ & $\mset{\seq t1k}$ & $s = \bigamp_{i=1}^k t_i$\\
  \bottomrule
\end{tabular}
\end{align*}
The semantic map $\sem-$ assigns the same value to $s$ as a state in the automaton as it does to $s$ as a variable in the corresponding system of equations, because the coinductive definitions of \S\ref{sec:semanticsofautomata} for automata and of \S\ref{sec:systemequations} for systems of equations mirror each other exactly.

It remains to show that a system of equations can be transformed to an equivalent set of expressions. We use a variant of Beki\'c's theorem \cite{Bekić1984}, which provides a general procedure to convert mutually recursive definitions to nested recursions on single variables. For example, the system of two equations in two variables
$x = e(x,y)$, $y = d(x,y)$
can be written
\begin{align*}
& x = \fix x{e(x,\fix y{d(x,y)})} && y = \fix y{d(\fix x{e(x,y)},y)}.
\end{align*}
As formulated in \cite{Bekić1984}, Beki\'c's theorem applies to least fixpoints of Scott-continuous functions on directed-complete partial orders, but we need a version that applies to unique fixpoints of contractive maps on complete metric spaces. Although the two variants rest on different assumptions, the procedure and result are the same. We give a formal proof of this variant.

\begin{lemmarep}
\label{lem:Bekic}
Let $m\ge 1$. The two systems of equations
\begin{align*}
x_i &= e_i,\ 1\le i\le m & x_i &= e_i[\fix{x_1}{e_1}/x_1],\ 1\le i\le m
\end{align*}
have the same solution.
\end{lemmarep}
\begin{proof}
By definition, the solutions of these two systems are environments $\rho',\rho'':\Var\to\DNS$ that are the unique fixpoints of contractive maps $\tau',\tau'' : (\Var\to\DNS) \to (\Var\to\DNS)$, respectively, with
\begin{align*}
\tau'(\rho)
&= \rho[\sem{e_i}\rho/x_i \mid 1\le i\le m] &
\tau''(\rho)
&= \rho[\sem{e_i[\fix{x_1}{e_1}/x_1]}\rho/x_i \mid 1\le i\le m]\\
&&&= \rho[\sem{e_i}\rho[\sem{\fix{x_1}{e_1}}\rho/x_1] / x_i \mid 1\le i\le m].
\end{align*}
Where the last inference follows by applying Lemma \ref{lem:substrebind}. Since $\rho''$ is a fixpoint of $\tau''$, we have
\begin{align}
\rho''
&= \rho''[\sem{e_i}\rho''[\sem{\fix{x_1}{e_1}}\rho''/x_1] / x_i \mid 1\le i\le m].\label{eq:rhopp}
\end{align}
In particular, by the semantic definition of fixpoint expressions in \S\ref{sec:semanticsofexpressions},
\begin{align*}
\rho''(x_1) = \sem{e_1}\rho''[\sem{\fix{x_1}{e_1}}\rho''/x_1] = \sem{\fix{x_1}{e_1}}\rho''.
\end{align*}
But then $\rho''[\sem{\fix{x_1}{e_1}}\rho''/x_1] = \rho''$, because it just rebinds $x_1$ to a value it already has. Thus we can rewrite \eqref{eq:rhopp} as
\begin{align*}
\rho''
&= \rho''[\sem{e_i}\rho''/x_i \mid 1\le i\le m]
= \tau'(\rho''),
\end{align*}
so $\rho''$ is also a fixpoint of $\tau'$. Since the fixpoint of $\tau'$ is unique, $\rho'' = \rho'$.
\end{proof}

\begin{theoremrep}
\label{thm:Bekic}
Given a system of equations $x_i = e_i$, $1\le i\le m$, its unique solution can be expressed as an $m$-tuple of expressions with no free occurrences of any $x_i$. 
\end{theoremrep}
\begin{proof}
By Lemma \ref{lem:Bekic} we can rewrite the system as $x_i = e_i[\fix{x_1}{e_1}/x_1],\ 1\le i\le m$. In the latter system, there are no free occurrences of $x_1$ except the defining occurrence $x_1 = e_1[\fix{x_1}{e_1}/x_1]$, which by the fixpoint axiom of Table \ref{tab:equations} is equivalent to the definition $x_1 = \fix{x_1}{e_1}$. But this definition is not needed in the residual system $x_i = e_i[\fix{x_1}{e_1}/x_1],\ 2\le i\le m$, as there are no free occurrences of $x_1$ in that system.

We now repeat with the system $x_i = e_i[\fix{x_1}{e_1}/x_1],\ 2\le i\le m$, substituting $\fix{x_2}{e_2[\fix{x_1}{e_1}/x_1]}$ for all free occurrences of $x_2$ throughout, and so on inductively. The final result is a fixpoint expression for $x_m$ with no free occurrence of $\seq x1m$.

By eliminating the variables in different orders, we can derive an expression for each variable in the same way.
\end{proof}

We have shown
\begin{corollary}
\label{cor:eqsolvable}
For every automaton $(S,\ell,\del)$, there is a closed expression $e_s$ for each state $s$ such that $\sem{e_s} = \sem{s}$.
\end{corollary}

Taken together, Theorem \ref{thm:expressiontoautomaton} and Corollary \ref{cor:eqsolvable} represent a full Kleene theorem for automata and expressions with probability and angelic nondeterminism.

\section{Coalgebraic Semantics}
\label{sec:coalgebras}

Every automaton can be rearranged to look like Fig.~\ref{fig:coalg}; that is, probabilistic states, followed by choice states, followed by action or terminal states, with the action states leading again to probabilistic states. One can use the finite distributive law $MD \to DM$ (corresponding to the axiom $(e_1\oplus_r e_2)\amp e_3 = (e_1\amp e_3)\oplus_r(e_2\amp e_3)$ of Table \ref{tab:equations}) to move probabilistic states before choice states, then the finite monad laws $MM\to M$ and $DD\to D$ to consolidate choice and probabilistic states, respectively. Dummy probabilistic and choice states with one successor can be added as necessary to maintain this structure.

In this form, the automaton becomes a coalgebra $(S,\del_S)$ for the functor $D(\naturals\times(\Mfin(-))^\Sigma)$ with structure map
$\del_S: S\to D(\naturals\times(\Mfin S)^\Sigma)$.
Here $\Mfin$ is the finite multiset functor; $\Mfin X$ is the set of multisets over $X$ with finite support and finite multiplicities. For a finite set $S$, $\Mfin S = \Ns$.

An element of $D(\naturals\times(\Ns)^\Sigma)$ represents a joint distribution on $\naturals\times(\Ns)^\Sigma$, where the first component represents the multiplicity with which the empty string is accepted and the remaining components represent multisets of states to which the automaton transitions after reading an input symbol. For example, in Fig.~\ref{fig:coalg}, if $s$ is the root of the diagram, then the marginal distribution of $\del_S(s)$ on the first component would give $2$, $1$, or $0$ with probabilities $p+r$, $q$, and $1-(p+q+r)$, respectively; the marginal distribution on the component corresponding to input symbol $a$ would give $\mset s$, $\mset{u,v}$, or $\mset{s,s,t}$ with probabilities $p+r$, $q$, and $1-(p+q+r)$, respectively; and the marginal distribution on the component corresponding to input symbol $b$ would give $\mset t$, $\mset{t,t}$, or $\mset{}$ with probabilities $1-(q+r)$, $r$, or $q$, respectively.

The space of behaviors of automata with actions $\Sigma$ is $\DNS$. This space forms an algebra $(\DNS,\eval)$ for the same functor with structure map
\begin{align*}
\eval : D(\naturals\times(\Mfin\DNS)^\Sigma) \to \DNS.
\end{align*}
The evaluation map $\eval$ is a composition of several steps, shown in the following diagram:
\begin{align*}
\begin{tikzpicture}[->, >=stealth', node distance=50mm, auto]
\small
\node (A) {$D(\naturals\times(\Mfin\DNS)^\Sigma)$};
\node (B) [right of=A] {$D(\naturals\times(D\Mfin(\NS))^\Sigma)$};
\node (C) [right of=B, node distance=45mm] {$D(\naturals\times\DNS^\Sigma)$};
\node (D) [below of=A, node distance=10mm, xshift=-8mm] {};
\node (E) [right of=D, node distance=29mm] {$DD(\naturals\times(\NS)^\Sigma)$};
\node (F) [right of=E, node distance=40mm] {$D(\naturals\times(\NS)^\Sigma)$};
\node (G) [right of=F, node distance=32mm] {$\DNS$};
\path (A) edge node {$\eval_1$} (B);
\path (B) edge node {$\eval_2$} (C);
\path (D) edge node {$\eval_3$} (E);
\path (E) edge node {$\eval_4$} (F);
\path (F) edge node {$\eval_5$} (G);
\end{tikzpicture}
\end{align*}
The first step $\eval_1$ applies the distributive law $\Mfin D\to D\Mfin$. The next step $\eval_2$ collapses a finite multiset of multisets in $\NS$ to a multiset in $\NS$ using the natural transformation $\Mfin M\to M$. Operationally, these two steps together can be viewed as applying the generalized version \eqref{eq:genamp} of $\amp$ to a finite multiset of elements of $\DNS$ to yield a single element of $\DNS$. The next step $\eval_3$ uses cartesian strength to move the inner $D$ to the outside. The next step $\eval_4$ uses the monad multiplication $DD\to D$. Operationally, these two steps together allow a two-step sampling process to be collapsed to a single sample. Finally, $\eval_5$ is the pushforward $Df$ of the bijection
\begin{align*}
& f:\naturals\times(\NS)^\Sigma\to\NS & f(n,\beta_a \mid a\in\Sigma)(x) &= \begin{cases}
n, & \text{if $x=\eps$}\\
\beta_a(y), & \text{if $x=ay$}
\end{cases}
\end{align*}
with inverse
\begin{align}
& f^{-1}:\NS\to\naturals\times(\NS)^\Sigma & f^{-1}(\alpha) &= (\alpha(\eps),\lam{x\in\Sigma^*}{\alpha(ax)} \mid a\in\Sigma).\label{eq:evalisom}
\end{align}

The coalgebra $(S,\del_S)$ and the algebra $(\DNS,\eval)$ of its behaviors work in concert according to the diagram
\begin{align}
\begin{array}c
\begin{tikzpicture}[->, >=stealth', auto]
\small
\node (NW) {$S$};
\node (NE) [right of=NW, node distance=75mm] {$\DNS$};
\node (SW) [below of=NW, node distance=12mm] {$D(\naturals\times(\Mfin S)^\Sigma)$};
\node (SE) [below of=NE, node distance=12mm] {$D(\naturals\times(\Mfin\DNS)^\Sigma)$};
\path (NW) edge node {$\sem{-}$} (NE);
\path (NW) edge node[swap] {$\del_S$} (SW);
\path (SW) edge node[swap] {$D(\id_\naturals\times(\Mfin\sem{-})^\Sigma)$} (SE);
\path (SE) edge node[swap] {$\eval$} (NE);
\end{tikzpicture}
\end{array}
\label{eq:final}
\end{align}
where $\sem-$ is the semantic map. This is a coalgebra/algebra diagram for the functor $D(\naturals\times(\Mfin(-))^\Sigma)$ and allows $\sem-$ to be defined uniquely by corecursion, as we now argue. We will show that the map
\begin{align*}
& \tau:(S\to\DNS)\to(S\to\DNS) & \tau(L) &= \eval\circ D(\id_\naturals\times(\Mfin L)^\Sigma)\circ\del_S
\end{align*}
on labelings $L:S\to\DNS$, in diagram form
\begin{align*}
\begin{tikzpicture}[->, >=stealth', auto]
\small
\node (NW) {$S$};
\node (NE) [right of=NW, node distance=75mm] {$\DNS$};
\node (SW) [below of=NW, node distance=12mm] {$D(\naturals\times(\Mfin S)^\Sigma)$};
\node (SE) [below of=NE, node distance=12mm] {$D(\naturals\times(\Mfin\DNS)^\Sigma)$};
\path (NW) edge node {$\tau(L)$} (NE);
\path (NW) edge node[swap] {$\del_S$} (SW);
\path (SW) edge node[swap] {$D(\id_\naturals\times(\Mfin L)^\Sigma)$} (SE);
\path (SE) edge node[swap] {$\eval$} (NE);
\end{tikzpicture}
\end{align*}
is contractive, thus has a unique fixpoint.

\begin{lemmarep}
\label{lem:nfsemantics}
$\tau$ is contractive with constant of contraction $1/2$.
\end{lemmarep}
\begin{proof}
Let $L:S\to\DNS$ be a labeling. Extend $L$ to domain $\Ns$ by defining
\begin{align*}
& \hat L:\Ns \to \DNS & \hat L(m) = {\amp}(\Mfin L(m)) = {\amp}\mset{L(t)\mid t\in m} = (\bigotimes_{t\in m}L(t)) \circ \Sigma^{-1}. 
\end{align*}

Let $s\in S$ and suppose
\begin{align*}
\del_S(s) = \sum_i r_i(n_i,m_{ia}) \mid a\in\Sigma).
\end{align*}
Applying $D(\id_\naturals\times(\Mfin L)^\Sigma)$ to $\del_S(s)$ yields 
\begin{align*}
\sum_i r_i(n_i,\Mfin L(m_{ia})) \mid a\in\Sigma).
\end{align*}
Applying $\eval$ to this yields
\begin{align}
\tau(L)(s)
&= \eval(\sum_i r_i(n_i,\Mfin L(m_{ia}) \mid a\in\Sigma))\nonumber\\
&= \eval_5(\eval_4(\eval_3(\eval_2(\eval_1(\sum_i r_i(n_i,\Mfin L(m_{ia}) \mid a\in\Sigma))))))\nonumber\\
&= \eval_5(\eval_4(\eval_3(\sum_i r_i(n_i,{\amp}\Mfin L(m_{ia}) \mid a\in\Sigma))))\nonumber\\
&= \eval_5(\eval_4(\eval_3(\sum_i r_i(n_i,\hat L(m_{ia}) \mid a\in\Sigma))))\nonumber\\
&= \eval_5(\sum_i r_i(\dirac{n_i} \otimes \bigotimes_{a\in\Sigma}(\hat L(m_{ia}))))\nonumber\\
&= \sum_i r_i ((\dirac{n_i}\otimes\bigotimes_{a\in\Sigma}\hat L(m_{ia}))\circ f^{-1}).\label{eq:phiapplied}
\end{align}
From the definition of the bijection \eqref{eq:evalisom}, we have
\begin{align}
f^{-1}([\alpha]_n) &= \set{(\beta(\eps),\lam x{\beta(ax)} \mid a\in\Sigma)}{\beta\equiv_n\alpha}\nonumber\\
&= \set{(\beta(\eps),\lam x{\beta(ax)} \mid a\in\Sigma)}{\beta(\eps)=\alpha(\eps)\wedge\forall a\in\Sigma\ \lam y{\beta(ay)}\equiv_{n-1}\lam y{\alpha(ay)}}\label{eq:fminusoneA}\\
&= \set{(\alpha(\eps),\lam x{\beta(ax)} \mid a\in\Sigma)}{\lam y{\beta(ay)}\in[\lam y{\alpha(ay)}]_{n-1},\ a\in\Sigma}\nonumber\\
&= \set{(\alpha(\eps),\gamma_a \mid a\in\Sigma)}{\gamma_a\in[\lam y{\alpha(ay)}]_{n-1},\ a\in\Sigma}\label{eq:fminusoneB}\\
&= \{\alpha(\eps)\}\times\prod_{a\in\Sigma}[\lam y{\alpha(ay)}]_{n-1}.\nonumber
\end{align}
The inference \eqref{eq:fminusoneA} follows from the argument
\begin{align*}
\beta\equiv_{n}\alpha\ &\Iff\ \forall x\in\Sigma^{\le n}\ \beta(x)=\alpha(x)\\
&\Iff\ \beta(\eps)=\alpha(\eps)\wedge\forall a\in\Sigma\ \forall y\in\Sigma^{\le n-1}\ \beta(ay)=\alpha(ay)\\
&\Iff\ \beta(\eps)=\alpha(\eps)\wedge\forall a\in\Sigma\ \lam y{\beta(ay)}\equiv_{n-1}\lam y{\alpha(ay)}.
\end{align*}
The inference \eqref{eq:fminusoneB} holds because $\beta\mapsto\lam y{\beta(ay)}$ is surjective; it is a split epimorphism with right inverse $a\cdot -$, the operation described in \S\ref{sec:injectivemonoidactions}. Applying \eqref{eq:phiapplied} to $[\alpha]_{n}$ yields
\begin{align*}
\tau(L)(s)([\alpha]_{n})
&= \sum_i r_i ((\dirac{n_i}\otimes\bigotimes_{a\in\Sigma}\hat L(m_{ia}))\circ f^{-1})([\alpha]_{n})\\
&= \sum_i r_i (\dirac{n_i}\otimes\bigotimes_{a\in\Sigma}\hat L(m_{ia}))(f^{-1}([\alpha]_{n}))\\
&= \sum_i r_i (\dirac{n_i}\otimes\bigotimes_{a\in\Sigma}\hat L(m_{ia}))(\{\alpha(\eps)\}\times\prod_{a\in\Sigma}[\lam y{\alpha(ay)}]_{n-1})\\
&= \sum_i r_i[\alpha(\eps)=n_i]\cdot\prod_{a\in\Sigma}\hat L(m_{ia})([\lam y{\alpha(ay)}]_{n-1})
\end{align*}
(the left-hand $[-]$ is the \emph{Iverson bracket}: $[\phi] = 1$ if $\phi$ is true, $0$ if false).
Now if $L_1\equiv_{n-1}L_2$, then by Lemma \ref{lem:integraldomain}, $\hat L_1\equiv_{n-1} \hat L_2$. Thus for all $i$ and $a$,
\begin{align*}
\hat L_1(m_{ia})([\lam y{\alpha(ay)}]_{n-1}) = \hat L_2(m_{ia})([\lam y{\alpha(ay)}]_{n-1}),
\end{align*}
which implies that $\tau(L_1)(s)([\alpha]_{n})=\tau(L_2)(s)([\alpha]_{n})$. As $s$ and $\alpha$ were arbitrary, $\tau(L_1)\equiv_n\tau(L_2)$.

We have shown that $L_1\equiv_{n-1}L_2$ implies $\tau(L_1)\equiv_n\tau(L_2)$. By Lemma \ref{lem:metricproperties}, $d(L_1,L_2)\le 2^{-n}$ implies $d(\tau(L_1),\tau(L_2))\le 2^{-(n+1)}$, thus $d(\tau(L_1),\tau(L_2))\le\frac 12d(L_1,L_2)$.
\end{proof}

\begin{theorem}
\label{thm:nfsemantics}
There is a unique semantic map $\sem-:S\to\DNS$ satisfying the diagram \eqref{eq:final}.
\end{theorem}
\begin{proof}
This follows from Lemma \ref{lem:nfsemantics} by the Banach fixpoint theorem.
\end{proof}

The coalgebra/algebra diagram \eqref{eq:final} plays the same role as a coalgebra diagram defining the unique coalgebra morphism to the final coalgebra for deterministic automata. Those diagrams can also be regarded as coalgebra/algebra diagrams, as the structure map of a final coalgebra is always invertible by Lambek's lemma \cite{Lambek68}. They give a unique semantic map in the same way, and this formulation explains why.

\subsection{Brzozowski Derivatives}
\label{sec:derivatives}

We would like a syntactic Brzozowski derivative $\Brz:\Exp\to D(\naturals\times(\Mfin\Exp)^\Sigma)$ that admits a coalgebra structure on expressions. It should satisfy the appropriate version of the diagram \eqref{eq:final}, to wit
\begin{align}
\begin{array}c
\begin{tikzpicture}[->, >=stealth', auto]
\small
\node (NW) {$\Exp$};
\node (NE) [right of=NW, node distance=75mm] {$\DNS$};
\node (SW) [below of=NW, node distance=12mm] {$D(\naturals\times(\Mfin\Exp)^\Sigma)$};
\node (SE) [below of=NE, node distance=12mm] {$D(\naturals\times(\Mfin\DNS)^\Sigma)$};
\path (NW) edge node {$\sem{-}$} (NE);
\path (NW) edge node[swap] {$\Brz$} (SW);
\path (SW) edge node[swap] {$D(\id_\naturals\times(\Mfin\sem{-})^\Sigma)$} (SE);
\path (SE) edge node[swap] {$\eval$} (NE);
\end{tikzpicture}
\end{array}
\label{eq:final2}
\end{align}

Intuitively, this means that
sampling from $\Brz(e)$, then applying the projection $\pi_\eps$ on the first component to obtain a number $n\in\naturals$ should give a random result distributed as if we had sampled $\sem e$, then asked for the multiplicity of $\eps$ in the resulting multiset; and sampling from $\Brz(e)$, then applying a projection $\pi_a$ for $a\in\Sigma$ to obtain a finite multiset of expressions $m_a$, then sampling $\sem{e'}$ independently for each expression $e'$ in $m_a$ and taking the multiset union of the results should give a random result distributed as if we had sampled $\sem e$ to get a multiset $\beta\in\NS$, then taken the usual Brzozowski derivative $\lam x{\beta(ax)}$ for weighted automata. Operationally,
\begin{align*}
\sample{\sem e}(\eps)\ &\sim\ \pi_\eps(\sample{(\Brz(e))})\\
\letin\beta{\sample{\sem e}}{\lam x{\beta(ax)}}\ &\sim\ \Sigma(\Mfin(\sample\circ\sem-)(\pi_a(\sample{\Brz(e)}))),\ a\in\Sigma,
\end{align*}
where the relation $\sim$ denotes that the left- and right-hand sides are identically distributed random variables. The following is the ``fundamental theorem'' for our system as described in \cite{Silva10}.
\begin{theoremrep}
There exists $\Brz:\Exp\to D(\naturals\times(\Mfin\Exp)^\Sigma)$ such that \eqref{eq:final2} commutes.
\end{theoremrep}
\begin{proof}
Given an expression $e$, use the following equations of \S\ref{sec:axioms} as reduction rules
\begin{align*}
& (e_1\amp e_2)\cmp e_3 \longrightarrow (e_1\cmp e_3)\amp(e_2\cmp e_3)\\
& (e_1\oplus_r e_2)\cmp e_3 \longrightarrow (e_1\cmp e_3)\oplus_r(e_2\cmp e_3)\\
& (e_1\cmp e_2)\cmp e_3 \longrightarrow e_1\cmp(e_2\cmp e_3)\\
& (e_1\oplus_r e_2)\amp e_3 \longrightarrow (e_1\amp e_3)\oplus_r(e_2\amp e_3)\\
& a \longrightarrow a\cmp\skp\\
& \skp\cmp e \longrightarrow e\\
& \fail\cmp e \longrightarrow \fail\\
& \fix xe \longrightarrow e\subst{\fix xe}x
\end{align*}
as necessary to transform $e$ to an equivalent expression $f$ satisfying the following grammar:
\begin{align*}
f &::= f_1 \oplus_r f_2 \mid g &
g &::= g_1 \amp g_2 \mid h &
h &::= a \cmp f \mid \skp \mid \fail
\end{align*}
The productivity assumption ensures that this is possible. For each maximal subexpression of the form $g$, consolidate all maximal subexpressions of the form $a\cmp e$ using associativity and commutativity of $\amp$ and the rule
\begin{align*}
& (a\cmp e_1)\amp(a\cmp e_2) \longrightarrow a\cmp(e_1\amp e_2)
\end{align*}
and use the rule
\begin{align*}
& e \longrightarrow e\amp (a\cmp\fail)
\end{align*}
as necessary so that each maximal subexpression of the form $g$ contains exactly one maximal subexpression of the form $a\cmp e$.

The resulting formula represents an entity of type $D(\naturals\times(M\,\Exp)^\Sigma)$ provably equivalent to the original formula $e$.
\end{proof}

\section{Conclusion}
\label{sec:future}

We have introduced a version of expressions and automata with probability and angelic nondeterminism modeled with multisets. Our main results are a full Kleene theorem asserting the equivalence of the two formalisms and a development of the corresponding coalgebraic theory, along with axioms and reasoning principles in both denotational and operational styles. To our knowledge, the Kleene theorem is the first result of its type for systems that combine probability and nondeterminism and speaks to the appropriateness of our approach. These results provide a foundation for KAT-style equational reasoning in systems that combine probability and nondeterminism.

Several questions remain that we have left for future work. While we have provided some rules for equational reasoning in \S\ref{sec:axioms}, we do not know whether the system is complete.

We have recently established that program equivalence is decidable (and that will be the subject of a forthcoming report), but the complexity is open. The problem is known to be PSPACE-hard, as the problem for Kleene algebra is PSPACE-complete \cite{KS96a}, but the greater expressiveness of probability and nondeterminism together indicates that the complexity of our system is likely higher.

An obvious next step is to add tests. In similar KAT-like systems, this typically entails no loss of efficiency for the decision problem and would allow the system to model the behavior of a simple probabilistic imperative programming language. Adding mutable variables as in \cite{GKM14a} would further increase expressiveness while likely preserving decidability.

Another question is that of automaton minimization. For deterministic finite automata, the Myhill-Nerode theorem provides a characterization of the minimal automaton for any regular language, as well as a procedure to generate one. We do not know of a similar bound for the automata of this paper. Information-theoretic lower bounds on the size of automata are less obvious than for deterministic automata, as probabilistic transitions allow complex behavior to be captured using very few states \cite{CSZ21}.

\begin{acks}
The support of the National Science Foundation (CCF-2008083) is gratefully acknowledged.
\end{acks}



\nosectionappendix
\appendix

\end{document}
\endinput